\documentclass[a4paper,USenglish,cleveref,thm-restate]{lipics-v2021}

\pdfoutput=1 
\hideLIPIcs  

\usepackage[appendix=append]{apxproof} 

\usepackage{amsmath, amssymb, amsthm, amsfonts}
\usepackage{xspace}
\usepackage{thm-restate}
\usepackage{misc}   
\usepackage{mathtools}

\usepackage{color}
\usepackage[dvipsnames]{xcolor}
\usepackage{graphicx}

\usepackage[shortcuts]{extdash} 

\usepackage{hyperref}
\usepackage{cleveref}

\usepackage{tikz}
\usetikzlibrary{calc}
\usepackage{tikz-cd}

\title{Let's Play Tag: Linear Time Evaluation of Conjunctive Queries under TGD Constraints}

\author{Nofar Carmeli}{Inria, LIRMM, University of Montpellier, CNRS,  France}{nofar.carmeli@inria.fr}{https://orcid.org/0000-0003-0673-5510}{}

\author{Carsten Lutz}{Institute of Computer Science, Leipzig University, Germany \and Center for Scalable Data Analytics and Artificial Intelligence (ScaDS.AI)}{clu@informatik.uni-leipzig.de}{https://orcid.org/0000-0002-8791-6702}{}

\author{Marcin Przyby{\l}ko}{Institute for Informatics, Warsaw University, Poland}{m.przybylko@uw.edu.pl}{https://orcid.org/0000-0003-1859-7055}{}

\authorrunning{N. Carmeli and C. Lutz and M. Przyby{\l}ko}

\Copyright{Nofar Carmeli and Carsten Lutz and Marcin Przyby{\l}ko}

\ccsdesc[500]{Theory of computation~Database query processing and optimization (theory)}

\keywords{Conjunctive Queries, Tuple Generating Dependencies, Query Evaluation, Answer Testing, Enumeration, Direct Access, Linear Time, Fine-Grained Complexity, Tagging.}

\category{}

\relatedversion{} 

\acknowledgements{We thank Karl Bringmann for helpful discussions.}

\nolinenumbers 

\EventEditors{John Q. Open and Joan R. Access}
\EventNoEds{2}
\EventLongTitle{30th International Conference on Database Theory (ICDT 2027)}
\EventShortTitle{ICDT 2027}
\EventAcronym{ICDT}
\EventYear{2027}
\EventLogo{}

\newtheoremrep{lemma}[theorem]{Lemma}
\newtheoremrep{prop}[theorem]{Proposition}
\newtheoremrep{thm}[theorem]{Theorem}

\newcommand{\core}{\mn{core}}

\newcommand{\atoms}{\mn{atoms}}

\newcommand{\loghyperclique}{loghyperclique}
\newcommand{\inc}{\subseteq}
\newcommand{\qbad}{q_{\mathrlap{\square}\angle}}
\newcommand{\Tbad}{\Tmc_{\mathrlap{\square}\angle}}
\newcommand{\qtri}{q_\triangle}
\newcommand{\qtriguard}{q_{\mathrlap{\bigcirc}\triangle}}
\newcommand{\Ttri}{T_\triangle}
\newcommand{\Ttriguard}{T_{\mathrlap{\bigcirc}\triangle}}
\newcommand{\Dtag}{D_{\mn{tag}}}
\newcommand{\Ttop}{T_{\mn{top}}}
\newcommand{\Tbot}{T_{\mn{bot}}}
\newcommand\inset[3]{\text{In}_{#1, #2}({#3})}
\newcommand\outset[3]{\text{Out}_{#1, #2}({#3})}
\newcommand\solleft[1]{\text{SolLeft({#1})}}
\newcommand\soltop[1]{\text{SolTop({#1})}}

\begin{document}

\maketitle

\begin{abstract}
  We study the limits of linear time evaluation of conjunctive queries under constraints expressed as tuple-generating dependencies (TGDs), across several modes of query evaluation: single-testing, all-testing, counting, lexicographic direct access, and enumeration. While full classifications seem far beyond reach, we propose an approach that,  for some evaluation modes and classes of TGDs, makes it possible to lift known dichotomies from the unconstrained setting. In particular, our approach applies to all mentioned evaluation modes except enumeration, when the constraints fall into one of~two classes: non-recursive sets of TGDs in which every TGD uses at most binary relation symbols in the head or has at most two frontier variables; and frontier-guarded full TGDs.
  We~further provide a collection of examples showcasing the challenges that arise for enumeration and for less restrictive classes of TGDs.
\end{abstract}

\section{Introduction}

Starting with a famous article of Yannakakis \cite{yannakakis-algotrithm}, there has been a continuous interest in~mapping out the potential and the limits of querying in linear time in data complexity.
Yannakakis established that 
the set of all answers to an acyclic full conjunctive query (CQ) can be computed in time linear in the size of the input and the output. The interest later shifted to other modes of query evaluation, such as enumeration, counting, and direct access, recently surveyed in~\cite{DBLP:journals/corr/abs-2506-17702}. Notably, it was shown in~\cite{bagan-enum-cdlin} that for conjunctive queries  that are acyclic and free-connex, answer enumeration is possible with linear preprocessing and constant delay.
The corresponding lower bounds
for self-join free CQs have been proved in \cite{bagan-enum-cdlin,BraultBaron}; see also \cite{berkholz-enum-tutorial}  for a~tutorial. The lower bounds are subject to the
hyperclique hypothesis from fine-grained complexity theory and to the
hypothesis that it is not possible (or at least thoroughly out of reach of today's methods) to multiply two $n \times n$-matrices in time $O(n^2)$.
 However, the former implies the latter~\cite{bringmann2025unbalanced}.
For~counting, analogous results that also pertain to CQs that are acyclic and free\=/connex have been obtained in \cite{mengel-counting-note}, building on prior work \cite{DBLP:journals/mst/0001M15,DBLP:conf/icalp/DellRW19}. The lower bounds additionally rely on the strong exponential time hypothesis~(SETH). We~remark
that, in~the case of counting, also the parameterized complexity
has been thoroughly investigated, the parameter being the size of the query; see for instance \cite{chen-counting,DBLP:conf/icalp/DellRW19}.

Given that databases are subject to various kinds of constraints, it is natural to ask about the limits of linear time querying in their presence as well. Note that posing constraints makes the class of admissible databases smaller, and thus intractable cases may turn into tractable ones, but not vice versa. The~existing results on this subject are much more sparse. The~dichotomy for the enumeration of self-join free CQs discussed above has been lifted to unary functional dependencies in \cite{carmeli-enum-func},
and the same has been done for direct access in~\cite{carmeli2023da}.
The~parameterized (but not linear time) complexity of single-testing under constraints presented by guarded tuple-generating dependencies (TGDs) is analyzed in \cite{barcelo_omq_limits-g}, linking it to tree-width. Here, single\=/testing means to decide, given a database and an answer candidate, whether the candidate is indeed an answer.
Likewise, the parameterized complexity of~counting query answers under guarded TGDs  was studied in \cite{LMCS-counting-23}. 

In this paper, we study the limits of linear time querying under TGD constraints across several modes of CQ evaluation: single-testing, all-testing, counting, lexicographic direct access, and enumeration. 
Since TGDs constitute a rather 
general formalism for specifying 
database constraints, which in fact was the original reason
for their introduction~\cite{DBLP:conf/icalp/BeeriV81},
it seems unrealistic to expect complete classifications. Indeed, 
 we~provide several examples showing that such classifications remain far beyond reach, for all the evaluation modes considered.  Instead,  we 
 propose and investigate an approach that sometimes makes it possible to lift dichotomies from the case without constraints to certain classes of TGDs: \emph{tagging}. In short, the idea of~tagging is to annotate database facts with query variables, thereby making explicit which query atom a fact is intended to realize.
 Tagging has been used before, for example in~\cite{bagan-enum-cdlin,carmeli2021ucqs,carmeli-enum-sjs,chen-counting}, to obtain results on the enumeration complexity of CQs with self-joins \cite{carmeli-enum-sjs}, among other applications. We combine tagging with other techniques, which can loosely be described as flooding and filtering, to make sure that the resulting database satisfies all TGDs. This approach does not work for all sets of TGDs, and we call a set \emph{tagging-friendly} if it does. We then prove that every non-recursive set of TGDs is tagging-friendly and so is every set of frontier-guarded full TGDs. In addition to bare tagging-friendlyness, we also require that tagged databases can be constructed within certain time bounds. For non-recursive sets of TGDs, this is possible if TGD heads use only relation symbols of maximum arity two or if every TGD has at most two frontier variables.

We use our approach to lift dichotomies for self-join-free CQs without constraints to~unrestricted CQs (i.e., possibly with self-joins) in the presence of constraints given by tagging-friendly sets of TGDs.
One proviso is that, in the statement of the lifted dichotomy, the structural properties of the original query
are replaced by those of its tagging companion. This is a query that is equivalent to the original one under the set of constraints, and it is constructed in the proof of tagging-friendliness. For instance, we show that a
CQ admits single-testing w.r.t.\ a set of tagging-friendly TGDs in linear time if and only if its tagging companion is acyclic, assuming the hyperclique hypothesis. We then consider all-testing, which asks for an algorithm that, after a preprocessing phase, repeatedly takes answer candidates as input and decides whether the candidate is indeed an answer. Lifting a dichotomy from \cite{DBLP:conf/pods/LutzP22}, our result is that all-testing is possible with linear preprocessing time and constant testing time if and only if the tagging companion is free-connex, assuming the hyperclique hypothesis.
We next lift the dichotomy for counting stated above and proceed
to consider lexicographic direct access where the CQ is endowed with a linear order $L$ of its answer variables, and we seek an algorithm that, after a preprocessing phase, repeatedly takes a positive integer $i$ as
input and outputs the $i$th answer according to the lexicographic ordering of all answers under~$L$ \cite{carmeli2023da}. Here, the relevant structural property is~that the tagging-companion is~acyclic, free\=/connex, and has no disruptive $L$-trio. As the last evaluation mode, we consider enumeration. While we are able to obtain dichotomies for two restricted cases, namely CQs of arity at most two and TGDs with only unary relations in the head, we also give examples showing that there are serious obstacles to obtaining more general results, even for very simple forms of (tagging-friendly) TGDs.

We remark that our results also remove the assumption of self-join freeness from the known dichotomies for CQs without constraints, presented for all-testing in \cite{DBLP:conf/pods/LutzP22}, for counting in~\cite{mengel-counting-note}, and for direct access in \cite{carmeli2023da}. While this might not be  surprising, we are not aware that it was made explicit before.

\section{Preliminaries}

{\bf Relational Databases.}
Fix countably infinite and disjoint
sets \Cbf, \Nbf, \Vbf of \emph{constants}, \emph{null constants},
and \emph{variables}. 
A {\em schema} \Sbf
is a finite set of relation symbols~$R$ with associated arity \mbox{$\mn{ar}(R) \geq 0$}. 
An {\em \Sbf-fact} is an expression of
the form $R(\bar c)$, where $R \in \Sbf$ and $\bar c$ is an
$\mn{ar}(R)$-tuple of constants from $\Cbf \cup \Nbf$.
An {\em $\Sbf$-instance} is a
set of \Sbf-facts,
and an {\em \Sbf-database} is a finite $\Sbf$-instance {that uses
only constants from \Cbf.} We write
$\dom(I)$ to denote the \emph{active domain} of instance $I$, that is, the set of constants used in  $I$.
A {\em homomorphism} from $I$ to an instance $J$ is a function $h :
\dom(I) \rightarrow \dom(J)$ such that $R(h(\bar c)) \in J$
for every $R(\bar c) \in I$.

\smallskip
\noindent
{\bf Conjunctive Queries.}
A {\em conjunctive query} (CQ)  over a schema $\Sbf$ takes
the form \mbox{$q(\bar x) \leftarrow \phi(\bar x, \bar y)$}
where $\bar x$ and $\bar y$ are tuples of variables and $\phi$ is a
conjunction of \emph{relational atoms} $R(\bar z)$ with
$R \in \Sbf$ and $\bar z$ a tuple of variables that occur in $\bar x$ or $\bar y$.
 Every variable from $\bar x$ must appear in some atom in $\phi$.
When we do not want to make $\phi(\bar x, \bar
y)$ explicit, we may denote $q(\bar x) \leftarrow \phi(\bar x, \bar
y)$ simply by $q$ or by $q(\bar x)$.
We call the variables in $\bar x$ 
\emph{answer variables}
and those in $\bar y$ \emph{existentially quantified} variables, and
use $\mn{var}(q)$  to denote the set of all variables
in $q$. 
The {\em arity} of  $q$
is the number of its answer variables, and $q$ is \emph{Boolean} if
it is of arity~0.  

Every CQ $q(\bar x)$ can  naturally be seen as a database $D_q$, known
as the {\em canonical database} of~$q$, obtained by viewing variables
as constants {from \Cbf}; answer variables and quantified variables are treated in the same way.
  A \emph{homomorphism}
$h$  from $q$ to an instance $I$ is a homomorphism from $D_q$ to
 $I$.
 A tuple $\bar c \in \dom(I)^{|\bar x|}$ is an {\em answer} to
 $q$ on $I$ if there is a homomorphism $h$ from $q$ to $I$ with
 $h(\bar x) = \bar c$.
We use  $q(I)$ to denote the  set of all answers to $q$ on~$I$.
For a CQ $q$, but also for any other syntactic object $q$, we use
$||q||$ to denote the number of symbols needed to write $q$ as
a word over a suitable alphabet.

 Let $q(\bar x)$ be a CQ. We say that $q(\bar x)$ is \emph{self\=/join free} if no relation symbol
occurs in more than one atom in it, and that it is \emph{full} if it contains no quantified variables. 
A \emph{join tree} for
$q(\bar x)$ is an undirected tree $T=(V,E)$ where $V$ is the set of
atoms in $\phi$ and for each variable $x \in \mn{var}(q)$, the set $\{
\alpha \in V \mid x \text{ occurs in } \alpha \}$ is a connected
subtree of $T$.  
Then, $q(\bar x)$ is \emph{acyclic} if it has a join
tree, and it is \emph{weakly acyclic} if it has a join tree with the connectedness condition applying only to existential variables $x$. 
We say that $q(\bar x)$ is \emph{free\=/connex} if the CQ that results from adding to $q$ an atom
$R(\bar x)$  `guarding' the answer variables, $R$ a
relation symbol of arity~$|\bar x|$, is acyclic.

\smallskip
\noindent
{\bf TGDs, Chase.}
A {\em tuple-generating dependency} (TGD) $T$ over $\Sbf$ is a
first-order sentence $ \forall \bar x \forall \bar y \,
\big(\phi(\bar x,\bar y) \rightarrow \exists \bar z \, \psi(\bar
x,\bar z)\big) $ such that $q_\phi(\bar x) \leftarrow 
\phi(\bar x,\bar y)$ and $q_\psi(\bar x) \leftarrow 
\psi(\bar x,\bar z)$ are CQs that do not contain constants.  We call
$\phi$ and $exists \bar z \, \psi$ the {\em body} and {\em head} of $T$ respectively.  The body may
be the empty conjunction, i.e.~logical truth, denoted by \mn{true}.
The variables in $\bar x$ are the
\emph{frontier variables}.
For simplicity, we write $T$ as $\phi(\bar x,\bar y) \rightarrow
\exists \bar z \, \psi(\bar x,\bar z)$. 
A TGD $T$ is {\em frontier-guarded} if its body is \mn{true} or contains a
\emph{guard atom} $\alpha$ that contains all  frontier variables~\cite{DBLP:conf/ijcai/BagetMRT11}. It is \emph{full} if there are no existentially quantified variables in the head. 
An instance $I$ over $\Sbf$
\emph{satisfies}~$T$, denoted $I \models T$, if $q_\phi(I) \subseteq
q_\psi(I)$. 
Let $\Tmc$ be a set of TGDs over a schema \Sbf. We use
$\preceq_{\Tmc}$ to denote the smallest (w.r.t.~`$\subseteq$')
reflexive and transitive binary relation on $\Sbf$ such that if
$\phi \rightarrow \psi$ is a TGD in $\Tmc$ with $R_1$ occurring in
$\phi$ and $R_2$ in $\psi$, then $R_1 \preceq_{\Tmc} R_2$.
We say that $\Tmc$ is \emph{non-recursive} if
$\preceq_{\Tmc}$ is antisymmetric, thus a partial order. 

Let \Tmc be a set of TGDs over a schema \Sbf. For CQs $q_1(\bar x_1)$
and $q_2(\bar x_2)$ over \Sbf that are of the same arity, we write
$q_1 \equiv_\Tmc q_2$ and say that $q_1$ and $q_2$ are \emph{equivalent w.r.t.} \Tmc if $q_1(D) = q_2(D)$ for all
\Sbf-databases $D$ that satisfy all TGDs in $T$.

The well-known chase procedure makes explicit in an instance the
consequences of a set of TGDs \cite{MaMS79,JoKl84, FKMP05,
  cali-taming-chase}.
There are different versions of the chase,
and we use the \emph{Skolem chase}. Let \Tmc be a set of TGDs over a schema \Sbf, let
 $T= \phi(\bar x,\bar y) \rightarrow \exists \bar z \, \psi(\bar x,\bar
z) \in \Tmc$ be a TGD, and $I$ an instance over~\Sbf. We assume a fixed Skolem function $f_{\psi}$,
range disjoint across all TGDs in~\Tmc, and such that $f_{\psi}(z)$ is a
fresh null for every variable $z$, that is,
$f_{\psi}(z) \notin \dom(I)$. For a tuple of variables $\bar z$, we write
$f_{\psi}(\bar z)$ to mean the result of 
replacing in $\bar z$ every variable $z$ 
with $f_{\psi}(z)$. 
We say that $T$ is \emph{applicable} in  $I$ to a 
tuple $\bar c$ over $\dom(I)$ 
if $\phi(\bar c,\bar c') \subseteq I$
for some tuple $\bar c'$ over $\dom(I)$,
but $\psi(\bar c,f_\psi(\bar z)) \not\subseteq I$.
In this case, the {\em result of applying $T$ in $I$ at $\bar c$} is the
instance $I \cup \psi(\bar c,f_\psi(\bar z))$.

A {\em chase sequence for $I$ with a set of TGDs $\Tmc$} is a sequence of instances
$I_0,I_1,\dots$ such that $I_0=I$ and each $I_{i+1}$ is the result of
applying some TGD from \Tmc at some tuple $\bar c$ of constants in $I_i$. The \emph{result} of this chase sequence is the
instance $J = \bigcup_{i \geq 0} I_i$. The chase sequence is
\emph{fair} if whenever a TGD $T \in \Tmc$ is applicable to a tuple
$\bar c$ in some $I_i$, then this application is part of the sequence. It is not difficult to show that every fair chase sequence
for $I$ with a set of TGDs $\Tmc$ leads to the same result, up to isomorphism.  We
denote this result with $\mn{ch}_\Tmc(I)$. For brevity, for a CQ $q$
we may write $\mn{ch}_\Tmc(q)$ in place of $\mn{ch}_\Tmc(D_q)$.

We call a set \Tmc of TGDs \emph{chase terminating} if for every database $D$,
$\mn{ch}_\Tmc(D)$ is finite. It is obvious that sets of full TGDs are chase terminating and well-known that also non-recursive sets of TGDs are. For a set of TGDs \Tmc that is chase terminating and a CQ $q$, $\mn{ch}_\Tmc(q)$ is finite and may  be viewed as a CQ with the same answer variables as $q$. It thus makes sense
to state the following, which is easy to prove.
\begin{lemma}
  \label{lem:chasingqueryequiv}
  If \Tmc is a set of TGDs that is chase terminating and $q$ is a CQ, then
  $q \equiv_\Tmc \mn{ch}_\Tmc(q)$.
\end{lemma}

\smallskip \noindent
{\bf Cores.}
A database $D$ is a \emph{core} if every endomorphism
on $D$ is injective, or equivalently, if every endomorphism is an automorphism.
\emph{The core} of a finite instance $D$ is the unique (up to isomorphism) instance
$\mn{core}(D) \subseteq D$ such that $D$ and $\mn{core}(D)$ are
homomorphically equivalent and  $\mn{core}(D)$ is 
a core. For a CQ $q$, we use
$\mn{core}(q)$ to denote the core of $D_q$, admitting however
only homomorphisms that are the identity on all answer variables.
\begin{lemmarep}
    \label{claim:cored-chase}
    If a database $D$ satisfies a TGD $T$,
    then so does  $\core(D)$. 
\end{lemmarep}
\begin{proof}
  Since $D$ and $\core(D)$ are homomorphically equivalent, there is a
  homomorphism $h_1$ from $\core(D)$ to $D$ and a homomorphism $h_2$
  from $D$ to $\core(D)$.  Assume that TGD
  $T=\phi(\bar x,\bar y) \rightarrow \exists \bar z \, \psi(\bar x,\bar
  z)$ is applicable to a pair $(\bar c,\bar c')$ in
  $\mn{core}(D)$.  Then, there is a homomorphism $g$ from $\phi$ to
  $\core(D)$ with $g(\bar x)=\bar c$ and $g(\bar y)=c'$. The
 composition $h_1 \circ g$ is a homomorphism from $\phi$ to $D$, and since $D \models \Tmc$, we find a $\bar d$ such that
  $\psi(h_1(\bar c'),\bar d) \in D$. This implies
  $\psi(h_2 \circ h_1(\bar c'),h_2(\bar d)) \in \core(D)$.  Now
  consider the homomorphism $h = h_2 \circ h_1$, which is an
  endomorphism on $\core(D)$ and thus an automorphism. Set
  $\bar c'' = h^{-1}(h_2(\bar d))$. Then, $h$ being an automorphism
  implies that $h^{-1}$ is an endomorphism, and thus
  $\psi(\bar c',\bar c'') \in \core(D)$ and  $T$ is
  satisfied at $(\bar c,\bar c')$.
\end{proof}
For a set of TGDs \Tmc and a CQ $q$,
a \emph{minimization of $q$ w.r.t.} \Tmc is
any CQ $q'$ that can be obtained from $q$ by repeatedly and exhaustively removing atoms, chosen arbitrarily, as long
as $q' \equiv_\Tmc q$ still holds. 
Clearly, for all minimizations $q_1,q_2$ of $q$ w.r.t.\ \Tmc we have $q_1 \equiv_\Tmc q_2$
since $q_1 \equiv_\Tmc q \equiv_\Tmc q_2$.
It is well-known that any minimization of a CQ $q$ w.r.t.\ the empty set of TGDs is exactly the core of $q$. Moreover,
it is easy to see that any minimization of $q$ w.r.t.\ a set of TGDs is a core (though not necessarily the core of $q$). A CQ that has no proper minimization w.r.t.\ \Tmc is said to be \emph{minimal w.r.t.} \Tmc.

\smallskip \noindent
{\bf Query Evaluation.}
We consider several modes of query evaluation: single-testing, all-testing, counting,  direct access in lexicographic order, and enumeration. Here, we define only the first two of them,  deferring the exact definition of the other modes to  subsequent  sections.
In all of the modes, a query evaluation algorithm for a CQ $q$ w.r.t.\ a set of TGDs~\Tmc, both over the same schema \Sbf, takes as an input an \Sbf-database $D$ that is promised to satisfy all TGDs in~$\Tmc$, and possibly additional inputs. 
We use \emph{data complexity} to measure the performance of an algorithm. \emph{Linear time} thus  means time $f(||q||+||\Tmc||) \cdot \Omc(||D||)$, where $f$ is a computable function, and \emph{constant time} depends only on $q$ and \Tmc, but not on $D$. 
For~an~input of size $n$, we use the word-RAM computation model with $O(\log(n))$-bit words and unit-cost operations. As a consequence, the values stored in the registers are bounded by~$n^c$ for some constant $c$, so we~can sort $m$ such values in time $O(c(m+n))$ using radix sort.

Like previous related results,
all hardness results in this article are \emph{conditional lower bounds}
in that they rely on the assumed hardness of a known problem. In particular, we shall use the \emph{hyperclique hypothesis}
which claims that, for all $k\geq 3$, it is impossible to determine the existence of a $k$-hyperclique in a $(k-1)$-uniform hypergraph with $n$ vertices in time $O(n^{k-1})$. We remark that this implies the \emph{triangle hypothesis},
which is the case $k=2$, 
stating that it is not possible to detect a triangle in an undirected graph with $n$ vertices in time $O(n^2)$. 
It also implies the \emph{BMM hypothesis} which states that two Boolean
$n \times n$ matrices cannot be multiplied in
time $O(n^2)$, see Proposition~5.16 of \cite{bringmann2025unbalanced}.

We now introduce single-testing
and all-testing. \emph{Single-testing} a query $q(\bar x)$ w.r.t.\ a set of TGDs \Tmc, both over the same schema \Sbf, is the problem to decide, given as input an \Sbf-database $D$ that satisfies all TGDs in \Tmc and a tuple $\bar c \in \dom(D)^{|\bar x|}$,  whether $\bar c \in q(D)$.
In \emph{all-testing} $q$ w.r.t.\ \Tmc,
the initial input is an \Sbf-database 
$D$ that satisfies all TGDs in \Tmc and there is a \emph{preprocessing phase}
in which the algorithm may produce suitable data 
structures, but no output,
and a \emph{testing phase}, where the algorithm repeatedly receives a tuple $\bar c$ as input and decides whether $\bar c \in q(D)$. The \emph{preprocessing time} of the algorithm is the time spent in the preprocessing phase, while the \emph{testing time} is the time it takes the algorithm to output a decision given a tuple during the testing phase. 

\section{Getting Started}\label{sec:pre-discussion}

We present examples that illustrate how TGDs can affect the complexity of query evaluation. We use Boolean CQs~$q$, for which all considered evaluation modes coincide and amount to deciding whether $q(D) \neq \emptyset$ or, equivalently, whether there is a homomorphism from $q$ to~$D$. We refer to this task as \emph{(Boolean) query evaluation}. For  Boolean CQs~$q$, this problem is known to be solvable in linear time if and only if $q$ is  acyclic, assuming the hyperclique hypothesis~\cite{yannakakis-algotrithm,BraultBaron,carmeli-enum-sjs}.

To get a first idea of the effect of adding TGD constraints, consider the (Boolean) query $\qtri\leftarrow R_1(x_1,x_2),R_2(x_2,x_3),R_3(x_3,x_1)$ which is cyclic and thus cannot be evaluated in
linear time on unrestricted databases. In contrast, $\qtri$ can be evaluated in linear time on databases that satisfy the TGD $\Ttri:R_1(x_1,x_2),R_2(x_2,x_3)\rightarrow R_3(x_3,x_1)$ because we can remove the atom $R_3(x_3,x_1)$  to obtain an  acyclic CQ $\qtri'$  with $\qtri \equiv_\Tmc \qtri'$.
This observation raises the   question
whether, given
a~CQ~$q$ and a~set of TGDs~\Tmc, there is a systematic way to construct  a CQ $q'$ such
that $q \equiv_\Tmc q'$ and the
complexity of evaluating $q$ over
databases that satisfy \Tmc is the same as the complexity of evaluating $q'$ over unrestricted databases. If this is possible, then it clearly allows us to adapt known results about the complexity of CQ evaluation on unrestricted databases to the setting with TGDs. 
The example may also suggest that we can always find such a $q'$ by removing
atoms. This, however, is not the case.
Consider, for instance, the 
TGD $\Ttriguard: R_1(x_1,x_2),R_2(x_2,x_3),R_3(x_3,x_1)\rightarrow S(x_1,x_2,x_3)$.
We can add to  $\qtri$ 
the atom $S(x_1,x_2,x_3)$  and obtain an acyclic CQ $\qtri''$  with $\qtri \equiv_\Tmc \qtri''$. Moreover, it is easy to see that removing atoms will not yield a CQ with these properties.

At this point, it might be tempting to think that there is simply a delicate question of choosing which atoms to keep, insert, and remove in order to turn a query into its most tractable form. Our first result, however, is that sometimes there is no equivalent acyclic CQ even though the problem  admits linear time evaluation. 

\begin{proposition}\label{prop:cyclic-easy-query}
There exists a Boolean CQ $\qbad$ and a set of TGDs $\Tbad$ such that:
\begin{enumerate}
    \item $\qbad$ admits linear time evaluation over databases that satisfy $\Tbad$; and
    \item there is no acyclic CQ $q$ such that $q(D)=\qbad(D)$ for every database $D$ that satisfies $\Tbad$.
\end{enumerate}
\end{proposition}

\Cref{prop:cyclic-easy-query} is proved by the following example.

\begin{figure}[t]
    \centering
    \begin{tikzpicture}[scale=.5]

        \node (x1) at (0, 0) {$x_1$};
        \node (x2) at (3,0)  {$x_2$};
        \node (x3) at (3,3)  {$x_3$};
        \node (x4) at (0,3)  {$x_4$};

        \draw [-latex] (x1)--(x2) node [midway, below] {\scriptsize $B$};
        \draw [-latex] (x1)--(x4) node [midway, left] {\scriptsize $L$};
        \draw [-latex] (x1)to node [midway, left] {\scriptsize $S$} (x3);
        \draw [-latex] (x3)--(x2) node [midway, right] {\scriptsize $R$};
        \draw [-latex] (x3)--(x4) node [midway, above] {\scriptsize $T$};

        \node (x1) at (5,0) {$x_1$};
        \node (x2) at (8,0)  {$x_2$};
        \node (x3) at (8,3)  {$x_3$};
        \node (x4) at (5,3)  {$x_4$};

        \draw [-latex] (x1)--(x2) node [midway, below] {\scriptsize $B$};
        \draw [-latex] (x1)--(x4) node [midway, left] {\scriptsize $L$};
        \draw [-latex] (x1)to node [midway, left] {\scriptsize $S$} (x3);
        \draw [-latex] (x3)--(x2) node [midway, right] {\scriptsize $R$};
        \draw [dashed, -latex] (x3)--(x4) node [midway, above] {\scriptsize $T$};

        \node (x1) at (10, 0) {$x_1$};
        \node (x2) at (13,0)  {$x_2$};
        \node (x3) at (13,3)  {$x_3$};
        \node (x4) at (10,3)  {$x_4$};

        \draw [dashed,-latex] (x1)--(x2) node [midway, below] {\scriptsize $B$};
        \draw [-latex] (x1)--(x4) node [midway, left] {\scriptsize $L$};
        \draw [-latex] (x1)to node [midway, left] {\scriptsize $S$} (x3);
        \draw [-latex] (x3)--(x2) node [midway, right] {\scriptsize $R$};
        \draw [-latex] (x3)--(x4) node [midway, above] {\scriptsize $T$};

    \end{tikzpicture}
    \caption{Illustration of \Cref{ex:Bool-cyclic-easy}: the query $\qbad$ and the rules $\Ttop$ and $\Tbot$.}
    
    \label{fig:example-weird-easy}
  \end{figure}

\begin{example}\label{ex:Bool-cyclic-easy}
    Let $\qbad() \leftarrow B(x_1,x_2), R(x_3,x_2), T(x_3,x_4), L(x_1,x_4), S(x_1,x_3)$, and let $\Tbad$ be the set  that contains  the TGDs
    \begin{itemize}
        \item $\Ttop:L(x_3,x_4), S(x_1, x_3), R(x_3,x_2), B(x_1,x_2) \to T(x_3,x_4)$
        \item $\Tbot: L(x_3,x_4), S(x_1, x_3), R(x_3,x_2), T(x_3,x_4) \to B(x_1,x_2)$.
    \end{itemize}
    The CQ and TGDs are depicted in \Cref{fig:example-weird-easy}, with the facts added by TGD heads displayed as dashed edges.
    To evaluate $\qbad$ in linear time on a database $D$,  first filter the relation $S$ to only contain facts that join with $L$ and $R$ as required, then for every fact $S(a,c)$ of~$S$, choose a single fact $L(a,d)$ and check whether $T(c,d)$ holds. The TGD $\Tbot$ guarantees that the query has an answer if the check succeeds, while the TGD $\Ttop$ guarantees that this check suffices, that is, no other $L$-facts need to be checked.
    Moreover, a careful  analysis of $\qbad$
    and $\Tbad$ shows that
    there is no acyclic CQ $q$
    with $\qbad \equiv_{\Tbad} q$.
Details are in the appendix.
\end{example}

\begin{toappendix}
We next prove that \Cref{ex:Bool-cyclic-easy} satisfies the requirements of \Cref{prop:cyclic-easy-query}.
First, observe that the query $\qbad$ is cyclic. In the presence of $\Tbad$, it can be reduced to one of two minimal queries: the one obtained by removing the atom $T(x_3,x_4)$, and the one obtained by removing the atom $B(x_1,x_2)$. Both minimal equivalent queries are cyclic. We next show  that any query equivalent to $\qbad$, even when taking $\Tbad$ into account, is cyclic. In the following, we write $D \models q$
to mean $q(D) \neq \emptyset$.
\begin{lemma}\label{prop:ex-no-cyclic-eq}
    There is no acyclic CQ $q'$
    such that for every database $D$ satisfying $\Tbad$, $q'(D)=\qbad(D)$.
\end{lemma}
\begin{proof}
    Notice that the canonical database $D_{\qbad}$ of $\qbad$, which satisfies $\Tbad$, contains only binary relations and no self-loops.
    Consider any query $q$ equivalent to $\qbad$ under $\Tbad$. Then
    $D_{\qbad} \models \qbad$
    implies  $D_{\qbad} \models q$, and consequently  $q$ contains only binary relations and no self-loops.
    Assume by way of contradiction that $q$ is acyclic. Clearly, any database that does not satisfy $\Tbad$ must contain a cycle or a self-loop. Consequently, $\Tbad$ is satisfied in the canonical database $D_q$ of~$q$. Together with $q \equiv_{\Tbad} \qbad$, this yields
    $D_q \models \qbad$. However, $\qbad$ is cyclic, and thus for  any acyclic database $D$ with only binary relations and no self-loops we must have $D \not\models \qbad$, contradicting the fact that $q$ is acyclic.
\end{proof}

As a consequence of \Cref{prop:ex-no-cyclic-eq}, we cannot get a linear time algorithm for evaluating $\qbad$ by solving an equivalent CQ. Instead, we show an ad-hoc algorithm for $\qbad$.

\begin{lemma}
    $\qbad$ admits linear time evaluation over databases satisfying $\Tbad$.
\end{lemma}
\begin{proof}
Consider the following algorithm. First compute $S'$, the set of $S$-facts that can be joined with $L$ and $R$ facts, i.e. $S' = \{S(a,c) \mid \exists b,d.\ L(d,a), R(c,b)\}$.
Then, for every fact $S(a,c)\in S'$, take a fact $L(a,d) \in D$, and check whether $T(c,d) \in D$. If the check is positive for some fact, return true. Otherwise, return false.

We can compute $S'$ in linear time using semi-joins, and we can look up specific facts in constant time per fact. Overall, the algorithm runs in linear time.

If the algorithm returns true, the check is positive, and we have found the facts $S(a,c)$, $L(a,d)$, and $T(c,d)$ in $D$. Since $S(a,c)\in S'$, there exists $b$ with $R(c,b)\in D$. Due to the rule $\Tbot$, we also have $B(a,b)\in D$, and so $D \models q$.

If $D \models q$, there exists a homomorphism from $q$ to $D$, and denote by $a,b,c,d$ the values it assigns $x_1,x_2,x_3,x_4$ respectively. Then, $S(a,c)\in S'$ by definition. Consider the fact $L(a,d')\in D$ chosen by the algorithm. Due to the rule $\Ttop$, we have $T(c,d')\in D$. Thus, the check will be positive for $S(a,c)$ and the algorithm will return true.
\end{proof}

\end{toappendix}

\Cref{ex:Bool-cyclic-easy} demonstrates that the approach suggested at the beginning of this section for determining the complexity of a CQ under a set of TGDs is not always applicable. In the remainder of this paper, we identify classes of TGDs for which it is. \Cref{ex:Bool-cyclic-easy}
rules out several candidates including the class of chase-terminating TGDs and the class of full TGDs.

\section{Tagging-Friendly TGDs}

We introduce the main technical tool used in this paper, tagging, which annotates each fact of a  database with  variables from the query while
also translating the database to a different schema. Its main purpose is to bridge the gap between unrestricted  databases and databases constrained by a set of TGDs. Simultaneously, it enables the transition from self-join-free CQs to CQs with self-joins.

For every CQ $q(\bar x)$ we define a \emph{self-join free version} $q_{\mn{sjf}}(\bar x)$,
 obtained  from $q$ by replacing every atom $R(\bar y)$ with $R_{\bar y}(\bar y)$ where $R_{\bar y}$ is a fresh relation symbol of the same arity as $R$. 
For tuples $\bar y=y_1,\dots,y_\ell$ and $\bar c= c_1,\dots,c_\ell$ of the same length, we use 
$\bar y \otimes \bar c$ to denote the tuple
$\langle y_1, c_1\rangle, \dots, \langle y_\ell,c_\ell\rangle$ obtained by component-wise
pairing. The naming scheme
for relation symbols just introduced plays a central role in the tagging construction.
\begin{definition}[Tagging]
\label{def:tagging}
    Let $q(\bar x)$ be a CQ over some schema \Sbf and $D$ a database over the schema of $q_{\mn{sjf}}$. The \emph{tagging of $D$ with $q$} is the database $D_{\mn{tag}}$ over \Sbf that contains, for every atom $R_{\bar y}(\bar y)$ in $q_{\mn{sjf}}$ and every fact $R_{\bar y}(\bar c)\in D$, the fact $R(\bar y \otimes \bar c)$.
\end{definition}
Note that constructing the database $D_{\mn{tag}}$ from  $D$ is possible in time linear in~$|D|$. 
\begin{lemma}
\cite[Lemma 5]{carmeli-enum-sjs}.
\label{lem:tagging}
    Let $q(\bar x)$ be a CQ that is a core, $D$  a database in the schema of $q_{\mn{sjf}}$, and $D_{\mn{tag}}$ the tagging of $D$ with $q$. Then, $\bar c \in q_{\mn{sjf}}(D)$ if and only if $\bar x \otimes \bar c \in q(D_{\mn{tag}})$.
\end{lemma}

We remark that $q(D_{\mn{tag}})$ in \Cref{lem:tagging} may additionally contain answers $\langle y_1, c_1\rangle, \dots, \langle y_n,c_n\rangle$ such that
$(y_1,\dots,y_n) \neq (x_1,\ldots,x_n)$. When using the lemma, we shall be
careful that these surplus answers
can simply be ignored without spoiling
the targeted time bounds.
To showcase the use of \Cref{lem:tagging}, we note that self-joins do not affect the complexity of single-testing and all-testing, without constraining TGDs. This can be seen as an extension of what was already known for Boolean CQs~\cite{carmeli-enum-sjs}, where single-testing and all-testing coincide. 
\begin{lemma}\label{cor:sjs}
    Let $q(\bar x)$ be CQ over some schema \Sbf that is a core. Then,
    \begin{enumerate}
        \item $q$ admits single-testing  in linear time if and only if $q_{\mn{sjf}}$ does;
        \item $q$ admits all-testing with linear preprocessing time and constant testing time if and only if $q_{\mn{sjf}}$ does.
    \end{enumerate}
\end{lemma}
\begin{proof}
   We treat Points~1 and~2 simultaneously. For the `if' directions,
   assume that we are given a database $D$ over \Sbf.
   We can construct in time linear
   in $|D|$ the database $D'=\{ R_{\bar y}(\bar c) \mid R(\bar c) \in D \text{ and } R(\bar y) \in q \}$ over the schema of $q_{\mn{sjf}}$. We then have 
    $q(D)=q_{\mn{sjf}}(D')$ and can test $q_{\mn{sjf}}$ on $D'$ in linear time.
   
   In the `only if' directions, we are given a database $D$ over the
   schema of $q_{\mn{sjf}}$. We may construct in  time linear in $|D|$
   the tagging $D_{\mn{tag}}$ of $D$ with $q$. By Lemma~\ref{lem:tagging}, $\bar c \in q_{\mn{sjf}}(D)$ if and only if $\bar x \otimes \bar c  \in q(D_{\mn{tag}})$. To test whether $\bar c \in q_{\mn{sjf}}(D)$, we can therefore test whether $\bar x \otimes \bar c \in q(D_{\mn{tag}})$, in linear time.
\end{proof}
\medskip
\noindent
\cref{cor:sjs} lifts known dichotomies for single-testing  and all-testing
\cite{DBLP:conf/pods/LutzP22} from self-join free CQs to the unrestricted case. For later use, we observe slightly stronger lower bounds that we would need here.
\begin{proprep}
\label{prop:singleallnoTGDs}
  Let $q(\bar x)$ be CQ that is a core. Then, assuming the hyperclique hypothesis,
    \begin{enumerate}
        \item $q$ admits single-testing  in linear time if it is weakly acyclic, and it does not admit single-testing in time $O(|D|+|\dom(D)|^2)$ otherwise;
        \item $q$ admits all-testing with linear preprocessing and constant testing time if  $q$ is free-connex, and it does not admit all-testing with preprocessing time $O(|D|+|\dom(D)|^2)$ and constant testing time otherwise.
    \end{enumerate}
\end{proprep}
\begin{proof}
  For all-testing, it is proved in \cite{DBLP:conf/pods/LutzP22} that (i)~every CQ $q(\bar x)$ that is free-connex admits all-testing
  with linear preprocessing and constant testing time, and (ii)~if $q$ is self-join free and not free-connex,  then it does not admit all-testing with linear preprocessing and constant testing time, assuming the hyperclique hypothesis. Point~(ii) is shown by reductions from the triangle hypothesis, the hyperclique hypothesis, and the BMM hypothesis, depending on the case. In the reductions from triangle detection and hyperclique detection, given an undirected (hyper)graph $G$ with $n$ vertices, one builds a database $D$ with $|\dom(D)| \in O(n)$. Likewise, in the reduction 
  from Boolean $n \times n$ matrix multiplication, one build a database $D$ with 
  $|\dom(D)| \in O(n)$.
  Based on this, it is easy to see that even an all-testing algorithm with preprocessing time $O(|D|+|\dom(D)|^2)$ and constant testing time would contradict the hyperclique hypothesis and the BMM hypothesis, respectively.
  It remains to apply Lemma~\ref{cor:sjs} and recall that the hyperclique hypothesis implies the triangle hypothesis and the BMM hypothesis.

  For single-testing, we start with the upper bound. It is 
  a classic result of Yannakakis that acyclic Boolean CQs can be evaluated in linear time \cite{yannakakis-algotrithm}. This result easily extends to acyclic Boolean CQs that that may use constants. Note that the definition of a join tree does not change in the presence of constants, and in particular the connectedness condition only applies to variables, but not to constants. What is more, there is an obvious linear time reduction from single-testing of a weakly acyclic CQ $q(\bar x)$ to the evaluation of  acyclic Boolean CQs: given a database $D$ and an answer candidate $\bar c$, let  $q'$ be obtained from $q$ by replacing the answer variables $\bar x$ with the  constants $\bar c$,  evaluate $q'$ over $D$, and return the result. Overall, we have thus seen that  single-testing acyclic CQs is possible in linear time.

Now for the lower bound.
Let $q(\bar x)$ be a CQ that is not weakly acyclic, and let $q'$ be the Boolean CQ obtained from $q$ by 
replacing every answer variable
with a distinct constant, giving rise to a candidate answer $\bar c$ to $q$. Then $q'$ is cyclic and consequently,
evaluating the self-join free version $q'_{\mn{sjf}}$ of $q'$
is not  possible in linear time, assuming the hyperclique hypothesis~\cite{BraultBaron}.
In the reduction from hyperclique detection, given a hypergraph with $n$ vertices, one  builds a database $D$ with $|\dom(D)| \in O(n)$. 
By Lemma~\ref{cor:sjs}, whose proof clearly extends to CQs with constants, the same is true for $q'$ itself.
What is more, there is an obvious linear time reduction
from the evaluation of $q'$ to
single-testing $q$. Consequently, a single-testing
algorithm for $q$ that runs in time $O(|D|+|\dom(D)|^2)$ would
contradict the hyperclique hypothesis.
\end{proof}
\medskip
\noindent
Let us now switch to the setting where TGD constraints are present.
The following definition summarizes the
properties that we would like a set of TGDs \Tmc to satisfy so that we can use tagging to clarify the complexity of evaluating CQs w.r.t.\ \Tmc.
\begin{definition}\label{def:taggingfriendly}
    Let \Tmc be a finite set of TGDs over a schema~\Sbf. 
    We say that \Tmc is \emph{tagging friendly} with time
bound~$t:\mathbb{N} \times
\mathbb{N} \rightarrow \mathbb{N}$ if for every CQ $q(\bar x)$ over \Sbf, there exists a CQ $q'(\bar x)$ such that $q \equiv_\Tmc q'$ and given a database $D$ over the schema of 
$q'_{\mn{sjf}}$, we can construct in time $t(|D|,|\dom(D)|)$
a database $D_{\mn{tag}}$ over \Sbf with   $\dom(D_{\mn{tag}}) \subseteq \mn{var}(q') \times \dom(D)$ such that
\begin{enumerate}

  \item $D_{\mn{tag}}$ satisfies all TGDs in \Tmc and 
        
  \item $\bar c \in q'_{\mn{sjf}}(D)$ if and only if $\bar x \otimes \bar c \in q(D_{\mn{tag}})$.

\end{enumerate}
We then call  $q'$ a \emph{tagging companion} of $q$ w.r.t.\ \Tmc.
\end{definition}
For brevity, we  say that \Tmc is tagging friendly \emph{with time bound} 
$O(f(|D|,|\dom(D)|))$ if $\Tmc$ is tagging friendly with 
some time bound $t(m,n) \in O(f(m,n))$, and that  \Tmc is tagging friendly \emph{with time bound}
$O(f(|D|))$ if $t$ is tagging friendly with 
some time bound $t(m,n) \in O(f(m,1))$.

\medskip

Note that Definition~\ref{def:taggingfriendly} does not insist that the database $D_{\mn{tag}}$ is constructed through tagging. However, the only way  for constructing 
$D_{\mn{tag}}$ that we are aware of involves tagging, combined with additional constructions that make sure that $D_{\mn{tag}}$ satisfies all TGDs in \Tmc. Note that, to achieve the latter, we cannot afford to chase with $\Tmc$ because the classes of TGDs that we are concerned with admit unrestricted CQs as the body, and thus checking whether a TGD  is applicable in a database $D$ requires to decide the existence of a homomorphism from a fixed, but unrestricted graph to $D$. In general, this destroys the targeted time bounds, which are preferably linear in $|D|$, but at most quadratic.

\begin{toappendix}

The following lemmas prepare for proving that, under
certain additional restrictions, sets of non-recursive TGDs are
tagging friendly. Let \Tmc be a set of TGDs and $q(x)$ a CQ. We refer to the facts in $\mn{ch}_\Tmc(q) \setminus D_q$ as \emph{chase facts}.
The lemmas shows that since we are using the Skolem chase,
(i)~$\Tmc$ is satisfied in $\mn{ch}_\Tmc(q)$ using only chase facts to witness TGD heads and (ii)~the answer $\bar x \in q(\mn{ch}_\Tmc(q))$ is only witnessed by homomorphisms that do not involve chase facts.
\begin{lemma}\label{lem:non-recursive-queries1}
    Let $\Tmc$ be a set of TGDs and  $q(\bar x_0)$  a CQ that is minimal w.r.t.~\Tmc. Then, for every TGD $\phi(\bar x,\bar y) \rightarrow
\exists \bar z \, \psi(\bar x,\bar z)$ in $\Tmc$ and every homomorphism $h$ from $\phi$ to $\mn{ch}_\Tmc(q)$, there is a homomorphism $g$ from $\psi$ to $\mn{ch}_\Tmc(q)\setminus D_{q}$ such that $h(x)=g(x)$ for all frontier variables $x$.
\end{lemma}
    \begin{proof}
    Let $h$ be a homomorphism from $\phi$ to $\mn{ch}_\Tmc(q)$. By definition of the (Skolem!) chase, there is then a homomorphism $g$ from $\psi$ to $\mn{ch}_\Tmc(q)$  such that
    all variables in $\bar z$ are mapped to nulls introduced by the chase. We claim that $g$ is also
    a homomorphism from $\psi$ to $\mn{ch}_\Tmc(q)\setminus D_{q}$. Let $R(\bar u)$ be an atom in $\psi$, so $R(g(\bar u))$ is in $\mn{ch}_\Tmc(q)$. If $u$ contains a variable from
    $\bar z$, then $g(u)$ is a null introduced by the chase.
    As a consequence, $R(g(\bar u))$ must be a chase fact,
    thus in $\mn{ch}_\Tmc(q)\setminus D_{q}$. Now
    assume that $u$ contains no variable from $\bar z$.
    Then, $R(g(\bar u))$ cannot be in $D_q$ because  $q$ is minimal w.r.t.\ \Tmc. Thus, $R(g(\bar u))$ is in $\mn{ch}_\Tmc(q)\setminus D_{q}$, as required.
 \end{proof}
\begin{lemma}\label{lem:non-recursive-queries2}
    Let $\Tmc$ be a set of non-recursive TGDs and  $q(\bar x_0)$  a CQ that is minimal w.r.t.~\Tmc. Then, 
        every homomorphism $h$ from $q$ to $\mn{ch}_\Tmc(q)$ such that $h(\bar{x}_0)=\bar{x}_0$ is also a homomorphism from $q$ to $D_{q}$.
\end{lemma}
    \begin{proof}
    Assume to the contrary that there is a homomorphism $h$ from $q$ to $\mn{ch}_\Tmc(q)$ such that  $h(\bar{x}_0)=\bar{x}_0$ and the image of $q$ under
    $h$ contains a chase fact. Let 
    $\Gamma$ be the set of relation symbols
    that occur in chase facts in the image of $q$ under
    $h$,  and choose an $R \in \Gamma$
    that is maximal among all symbols in $\Gamma$
    regarding the order `$\preceq_\Tmc$'
    on relation symbols induced by the non-recursive set of TGDs \Tmc; in other words: there is no $S \in \Gamma$ such that $R \preceq_\Tmc S$.
    Then, there is no chase fact in the image of $q$ under
    $h$  whose
    derivation in the chase of $q$ with \Tmc depends on an atom that uses $R$. 

Assume that the number of $R$-atoms in $q$
is $n$. Then, there are also $n$ $R$-facts in the image of $q$ under
    $h$. We know that at least one of these $R$-facts is a chase fact, and thus at most $n-1$ of the $R$-facts in the image of $q$ under
    $h$ are atoms in $q$. Thus, there must be some atom $R(\bar z)$ in $q$ that is not in the range of $h$. Let $q'$ denote $q$ without~$R(\bar z)$. 
Since there is no chase fact in the image of $q$ under
    $h$  whose derivation depends on an $R$-atom in $q$, all facts in the range of 
$h$ are in  $\mn{ch}_\Tmc(q')$.
    Thus, $h$ is also a homomorphism from $q$ to $\mn{ch}_\Tmc(q')$. This means that $q \equiv_\Tmc q'$, contradicting the minimality of $q$ w.r.t.\ \Tmc.
 \end{proof}
\end{toappendix}

We next show that, under
certain additional restrictions, sets of non-recursive TGDs are
tagging friendly. To construct $D_{\mn{tag}}$, we use tagging and then `flood' with facts in a careful way to achieve that all TGDs are satisfied.
\begin{lemmarep}\label{lem:non-recursive-construction}
    Let \Tmc be a finite set of non-recursive TGDs over schema \Sbf, and let $k \geq 0$.
    Then, \Tmc is tagging friendly with time bound $O(|D|+|\dom(D)|^k)$
     if all TGDs in
    \Tmc satisfy one of the following conditions:
    \begin{enumerate}
    
        \item relation symbols in the head are of arity at most $k$;
        
        \item there are at most $k$ frontier variables.
        
    \end{enumerate}
    Moreover, for every CQ $q(\bar x)$ over \Sbf, any minimization
    of $q$ w.r.t.\ \Tmc is a tagging companion of $q$ w.r.t.~\Tmc. 
\end{lemmarep}
\begin{proof}
  Let \Tmc be a finite set of non-recursive TGDs over schema \Sbf, and let $k \geq 0$. We first prove the lemma under the assumption that \Tmc satisfies Condition~1
  and later sketch the modification required for 
  Condition~2.
Let $q(\bar x)$ be a CQ over \Sbf, and let
$q'(\bar x)$  be a minimization of $q$ w.r.t.\ \Tmc. 
Assume that we are given a database 
$D$ over the schema of $q'_{\mn{sjf}}$. We
have to show that we can construct in time $O(|D|+|\dom(D)|^k)$ a database $D_{\mn{tag}}$ over \Sbf that satisfies Conditions~1 and~2 from~\Cref{def:taggingfriendly}. We first
construct the tagging $D'_{\mn{tag}}$ of $D$
with $q'$ in time $O(|D|)$. Since $q'$ is
minimized w.r.t.~\Tmc, it is a core, and
thus \Cref{lem:tagging} yields $\bar c \in q'_{\mn{sjf}}(D)$ iff $\bar x \otimes \bar c \in q'(D'_{\mn{tag}})$ for all $\bar c \in \dom(D)^{|\bar x|}$.

The database $D'_{\mn{tag}}$ is not yet as desired because it does not necessarily satisfy the TGDs in~$\Tmc$. Since
we cannot chase with \Tmc while achieving the desired time bound, we  resort to the brute force approach of flooding all relevant relations with all possibly relevant facts.
With $\mn{atoms}(q)$, we denote the set of atoms in a query $q$.
Define
\[D_{\mn{tag}} = D'_{\mn{tag}} \cup \{R(\bar{y} \otimes \bar{c}) \mid R(\bar y)\in\atoms(\mn{ch}_\Tmc(q')) \setminus\atoms(q') \text{ and } \bar{c}\in\dom(D)^{|\bar y|}\}.\]

Recall that $q'$ is fixed, so the number of $R(\bar{y})$ atoms considered in this construction is a constant. Moreover, we assume that Condition~1 holds and thus the relation symbols in $\atoms(\mn{ch}_\Tmc(q')) \setminus\atoms(q')$ have arity at most $k$. Consequently, $D_{\mn{tag}}$ can be computed from $D'_{\mn{tag}}$ within $O(|\dom(D)|^k)$ time.
The overall construction time of $D_{\mn{tag}}$ from $D$ is thus $O(|D|+|\dom(D)|^k)$, as required. 

We make the following observation:
\begin{quote}

   ($*$) 
    if $R(\bar y \otimes c) \in D_{\mn{tag}}$, then $R(\bar y)$ is an atom in $\mn{ch}_\Tmc(q')$.
    
\end{quote}
In fact, $R(\bar y \otimes c)$ may already  be a fact in  
$D'_{\mn{tag}}$. Then, $R(\bar y)$ is an atom in
$q'$ by construction of $D'_{\mn{tag}}$. Otherwise, $R(\bar y \otimes c)$ was added during the construction of
$D_{\mn{tag}}$. Then, $R(\bar y)$ is an atom in  $\mn{ch}_\Tmc(q') \setminus q'$. In either case,
$R(\bar y)$ is an atom in $\mn{ch}_\Tmc(q')$.

We now show that $D_{\mn{tag}}$ satisfies all TGDs in $\Tmc$. Let  $T = \phi(\bar u,\bar y) \rightarrow
\exists \bar z \, \psi(\bar u,\bar z)$  be such a TGD.
If there is a homomorphism $h$ from $\phi$ to $D_{\mn{tag}}$, then by ($*$), $g = \pi_1 \circ h$, where $\pi_1$ is the projection to the first component, is a homomorphism  from $\phi$ to $\mn{ch}_\Tmc(q')$. Thus, by \Cref{lem:non-recursive-queries1}, there is also a homomorphism $g'$ from $\psi$ to the \emph{chase facts} of $\mn{ch}_\Tmc(q')$ that agrees with $g$ on the frontier variables. Since the construction of $D_{\mn{tag}}$ pairs all chase facts from  $\mn{ch}_\Tmc(q')$ with all elements of $\dom(D)$, this implies a homomorphism $h'$ from  $\psi$ to $D_{\mn{tag}}$ that agrees with $h$ on the frontier variables $\bar u$, as required for the satisfaction of $T$. More precisely, we may define $h'$ by choosing any $c \in \dom(D)$ and (i)~setting $h'(u)=h(u)$ for all frontier variables $u$ in $\bar u$ and (ii)~$h'(z)=(g'(z),c)$ for all
variables $z$ in $\bar z$.

We next argue that $\bar x \otimes \bar c \in q'(D'_{\mn{tag}})$ if and only if $\bar x \otimes \bar c \in q'(D_{\mn{tag}})$ for all $\bar c \in \dom(D)^{|\bar x|}$.
The `only if' direction is immediate since
$D'_{\mn{tag}} \subseteq D_{\mn{tag}}$
 and thus $q'(D'_{\mn{tag}})\subseteq q'(D_{\mn{tag}})$. For the `if' direction, consider an answer $\bar x \otimes \bar c \in q'(D_{\mn{tag}})$. Then, there is a homomorphism
 $h$ from $q'$ to $D_{\mn{tag}}$ with $h(\bar x)=\bar x \otimes \bar c$. By ($*$),
 $g = \pi_1 \circ h$ is a homomorphism from $q'$ to $\mn{ch}_\Tmc(q')$ with $g(\bar x)=\bar x$. According to \Cref{lem:non-recursive-queries2}, the image of $q'$ under $g$ contains no chase facts. But since all facts $R(\bar y \otimes \bar d) \in D_{\mn{tag}} \setminus D'_{\mn{tag}}$ are such that $R(\bar y)$
 is a chase fact in $\mn{ch}_\Tmc(q')$, this
 means that $h$ is actually a homomorphism from
 $q'$ to  $D'_{\mn{tag}}$. Consequently $\bar x \otimes \bar c \in q'(D'_{\mn{tag}})$.

From $q \equiv_\Tmc q'$ and since $D_{\mn{tag}}$
satisfies all TGDs in \Tmc, we further get $q'(D_{\mn{tag}})=q(D_{\mn{tag}})$.
Overall, we obtain $\bar c \in q'_{\mn{sjf}}(D')$ if and only if $\bar x \otimes \bar c \in q(D_{\mn{tag}})$, as required.

\medskip

The proof remains essentially identical when Condition~2 of Lemma~\ref{lem:non-recursive-construction}
is assumed in place of Condition~1, that is, when there are at most $k$ frontier variables. The only difference lies
in the construction of the database $D_{\mn{tag}}$, which is 
now defined as 
$$
\begin{array}{r@{\;}l}
D_{\mn{tag}} = D'_{\mn{tag}} \cup \{R(\bar{y} \otimes \bar{c}) \mid &R(\bar y)\in\atoms(\mn{ch}_\Tmc(q')) \setminus\atoms(q') \text{ and } \\[1mm]& \bar{c}\in\dom(D)^{|\bar y|} \text{ uses only $k$ distinct constants} \}.
\end{array}
$$
Let $r$ be the maximum arity of relation symbols in \Sbf. Note that
for our purposes, $r$ is a constant. It is easy to see that the 
number of tuples $\bar c$ considered in the definition of $D_{\mn{tag}}$
is bounded by $|\dom(D)|^k \cdot k^r$.
As a consequence,  $D_{\mn{tag}}$ can again be constructed in time $O(|D|+|\dom(D)|^k)$.

One can show in the same way as before that $D_{\mn{tag}}$ satisfies
all TGDs in \Tmc. In particular, the crucial step at the end of concluding that there is a homomorphism $h'$ from $\psi$ to $D_{\mn{tag}}$
that agrees with $h$ on the frontier variables $\bar x$ still goes
through: instead of choosing any $c \in \dom(D)$, we now choose such a $c$ such that it occurs in a second component of some  pair $h(u)$, for a frontier variable $u$. If the TGD has no frontier variables, we again choose $c \in \dom(D)$ arbitrarily.
\end{proof}
\cref{lem:non-recursive-construction} only turns out to be useful for us in
the case that $k \in \{1,2\}$ as otherwise, the high time bound is prohibitive.
We next observe that sets of TGDs that are both frontier-guarded and full are also tagging friendly. Here, we construct $D_{\mn{tag}}$
by first filtering out from $D$ certain facts that cannot
possibly be involved in matches of $q$, and then tagging. The former serves to ensure that all
TGDs are satisfied.
\begin{lemmarep}\label{lem:guarded-friendly}
    Let \Tmc be a finite set of frontier-guarded full TGDs over  schema \Sbf. Then, \Tmc is tagging friendly with time bound $O(|D|)$, and for every CQ $q(\bar x)$ over \Sbf, the core of $\mn{ch}_\Tmc(q)$ is a tagging companion of $q$ 
    w.r.t.~\Tmc.
\end{lemmarep}
\begin{proof}
Define $q'$ to be the core of $\mn{ch}_\Tmc(q)$, and assume we are given a database $D$ over the schema of $q'_{\mn{sjf}}$. We
show that we can construct in time $O(|D|)$ a database $D_{\mn{tag}}$ over \Sbf that satisfies \Tmc such that $\bar c \in q'_{\mn{sjf}}(D)$ if and only if $\bar x \otimes \bar c \in q(D_{\mn{tag}})$.

We first construct from $D$
a database $D_0$ by filtering out certain
facts that are irrelevant for answering $q'_{\mn{sjf}}$. This step will ensure that
the database $D_{\mn{tag}}$ ultimately
constructed satisfies all TGDs in~$\Tmc$. More
precisely, $D_0$ is obtained from $D$ by removing every fact $R_{\bar y}(\bar a)$ such that either:
\begin{itemize}
    \item[(i)] there is no substitution mapping $\bar y$ to $\bar a$; or
    \item[(ii)] there is such a substitution $\sigma$, and there is an atom $S_{\bar z}(\bar{z})$ in $q'_{\mn{sjf}}$ with ${\bar z}\subseteq{\bar y}$ (every variable in $\bar z$ occurs in $\bar y$) such that $S_{\bar z}(\sigma(\bar{z}))\notin D$. 
    
\end{itemize}
This filtering step does not change answers to $q'_{\mn{sjf}}$, that is, $q'_{\mn{sjf}}(D)=q'_{\mn{sjf}}(D_0)$. The `$\supseteq$' direction is immediate. For
`$\subseteq$', let $h$ be a 
homomorphism from $q'_{\mn{sjf}}$ to $D$.
Then, for every atom $R_{\bar y}(\bar y)$ of $q'_{\mn{sjf}}$, the fact $R_{\bar y}(h(\bar y))$ is in $D$. So, $h$ is a substitution from $\bar y$ to $h(\bar y)$, and this fact is not removed in filtering step (i). 
Since $q'_{\mn{sjf}}$ is self-join-free, the only atom using $R_{\bar y}$ is $R_{\bar y}(\bar y)$, and $h$ is the only substitution mapping it to $R_{\bar y}(h(\bar y))$. Recall that for every atom $S_{\bar z}(\bar{z})$ in $q'_{\mn{sjf}}$, we have that $S_{\bar z}(h(\bar{z}))\in D$, so $R_{\bar y}(h(\bar y))$ is not removed in the second filtering step either. We conclude that $R_{\bar y}(h(\bar y))\in D_0$, and $h$ is a 
homomorphism from $q'_{\mn{sjf}}$ to $D_0$.

\smallskip
The desired database $D_{\mn{tag}}$ over \Sbf is the tagging 
of $D_0$ with $q'$. Clearly, $D_{\mn{tag}}$ can be constructed
in time $O(|D|)$, as required.
By \Cref{lem:tagging}, $\bar c \in q'_{\mn{sjf}}(D_0)$ if and only if $\bar x \otimes \bar c \in q'(D_{\mn{tag}})$. Since we also have that $q'_{\mn{sjf}}(D)=q'_{\mn{sjf}}(D_0)$ and $q \equiv_\Tmc q'$, we get that $\bar c \in q'_{\mn{sjf}}(D)$ if and only if $\bar x \otimes \bar c \in q(D_{\mn{tag}})$, as desired.
It is left to show that $D_{\mn{tag}}$ satisfies \Tmc.

Since the TGDs are full, we may assume w.l.o.g.\ that they have a single head atom.
Consider a TGD $\phi(\bar y,\bar z) \rightarrow  S(\bar y)$ in \Tmc.
Assume there exists a homomorphism $h$ from $\phi$ to $\Dtag$; we need to show that $S(h(\bar{y}))\in \Dtag$.

 Recall that $\pi_i$ denotes the projection to the $i$th component.
 First, we show that $q'_{\mn{sjf}}$ contains the atom $S_{\pi_1\circ h(\bar y)}(\pi_1\circ h(\bar y))$.
Consider an atom $R(\bar v) \in \phi$.
We have that $R(h(\bar{v}))\in\Dtag$.
By construction of $\Dtag$, 
the atom $R_{\pi_1\circ h(\bar v)}(\pi_1\circ h(\bar v))$ appears in $q'_{\mn{sjf}}$.
By definition of $q'_{\mn{sjf}}$, the atom $R(\pi_1\circ h(\bar v))$ appears in $q'$, and since $q'\subseteq \mn{ch}_\Tmc(q)$, it appears in $\mn{ch}_\Tmc(q)$ as well. 
Since this is true for every atom $R(\bar v) \in \phi$ and since $\mn{ch}_\Tmc(q)$ satisfies \Tmc, we conclude that $\mn{ch}_\Tmc(q)$ contains the atom $S(\pi_1\circ h(\bar y))$.
Since $q'$ is a core of $\mn{ch}_\Tmc(q)$,
it is also a \emph{retract}, implying that
there is a homomorphism from 
$\mn{ch}_\Tmc(q)$ to $q'$ that is the 
identity on all variables in $q'$~\cite{DBLP:books/daglib/0013017}.
As we saw that $R(\pi_1\circ h(\bar v))$ appears in $q'$ for every atom $R(\bar v) \in \phi$, we have that $\pi_1\circ h(\bar y)$ appear in $q'$, and so the atom $S(\pi_1\circ h(\bar y))$ appears in $q'$ as well.
Finally, by definition of $q'_{\mn{sjf}}$, it contains the atom $S_{\pi_1\circ h(\bar y)}(\pi_1\circ h(\bar y))$.

Since the TGD is frontier-guarded, $\phi$ contains an atom $G(\bar u)$  with $\bar y\subseteq \bar{u}$.
Since $h$ is a homomorphism from $\phi$ to $\Dtag$, we get that $G(h(\bar u))\in\Dtag$.
By construction of $\Dtag$, we have that $G_{\pi_1\circ h(\bar u)}(\pi_2\circ h(\bar u))\in D_0$ and $q'_{\mn{sjf}}$ contains the atom $G_{\pi_1\circ h(\bar u)}(\pi_1\circ h(\bar u))$.
Since $G_{\pi_1\circ h(\bar u)}(\pi_2\circ h(\bar u))$ was not deleted from $D_0$ during filtering step (i), there is a substitution mapping ${\pi_1\circ h(\bar u)}$ to $\pi_2\circ h(\bar u)$. Since it was not deleted during filtering step (ii), and since $q'_{\mn{sjf}}$ contains the atom $S_{\pi_1\circ h(\bar y)}(\pi_1\circ h(\bar y))$ with $\pi_1\circ h(\bar y)\subseteq \pi_1\circ h(\bar{u})$, we conclude that $D_0$ contains the fact $S_{\pi_1\circ h(\bar y)}(\pi_2\circ h(\bar y))$.
By definition of $\Dtag$, we get that it contains the fact $S(h(\bar y))$, so the TGD is satisfied.
\end{proof}
Note that \cref{lem:guarded-friendly}
establishes a stronger time bound
than \cref{lem:non-recursive-construction} except for the case $k=1$.
Nevertheless, the weaker bound of  $O(|D|+|\dom(D)|^2)$ suffices to establish 
the desired dichotomy results in 
 subsequent sections, apart from one case.

\section{Single-Testing and All-Testing}

 The following theorem lifts \Cref{prop:singleallnoTGDs} to sets of TGDs that are tagging
friendly with time bound $O(|D|+|\dom(D)|^2)$.
\begin{theorem}\label{thm:testing-dichotomies}
Let \Tmc be a finite set of TGDs  that is tagging friendly with time
bound $O(|D|+|\dom(D)|^2)$.
Further, let $q(\bar x)$ be a CQ over the same schema as \Tmc, and let $q'(\bar x)$ be its tagging companion.
Then, assuming the hyperclique hypothesis,
    \begin{enumerate}
        \item $q$ admits single-testing  w.r.t.\ $\Tmc$ in linear time if and only if $q'$ is  weakly acyclic;
        \item $q$ admits all-testing w.r.t.\ $\Tmc$ with linear preprocessing and constant testing time if and only if $q'$ is free-connex.
    \end{enumerate}
\end{theorem}

\begin{proof}
Since $q \equiv_\Tmc q'$, we may consider $q'$ in place of $q$ when proving
the theorem. 
The upper bounds are thus simply inherited
from the case without TGDs. For the lower bound in Point~1, assume that $q'$ is not weakly acyclic. Then, the same is true for the self-join-free version $q'_{\mn{sjf}}$ of $q'$.
Thus, by \Cref{prop:singleallnoTGDs},  $q'_{\mn{sjf}}$ does not admit single-testing over unrestricted databases in  time $O(|D|+|\dom(D)|^2)$, assuming the hyperclique hypothesis.
Given a database $D$ in the schema of $q'_{\mn{sjf}}$, 
since $q'$ is a tagging companion of~$q$, 
we can construct in time $O(|D|+|\dom(D)|^2)$ a database $\Dtag$ for $q$ that satisfies all TGDs in $\Tmc$ and such that $\bar c \in q'_{\mn{sjf}}(D)$ if and only if $\bar x \otimes \bar c \in q(D_{\mn{tag}})$. 
Thus, a linear time single-testing algorithm for $q$ w.r.t.\ $\Tmc$ would imply an $O(|D|+|\dom(D)|^2)$ single-testing algorithm for $q'_{\mn{sjf}}$ over unrestricted databases, which does not exist.
The lower bound in 
Point~2 is analogous.
\end{proof}

As an important special case, \Cref{cor:non-rec-Bool} identifies the condition for linear time evaluation of Boolean CQs. Recall that for such CQs,  single-testing and all-testing all collapse
into evaluation. In addition, acyclicity, weak acyclicity, and free-connexity all coincide.

\begin{corollary}\label{cor:non-rec-Bool}
Let \Tmc be a finite set of TGDs that is tagging friendly with time
bound $O(|D|+|\dom(D)|^2)$.
Further, let $q$ be a Boolean CQ over the same schemas as \Tmc, and let $q'$ be its tagging companion.
Then, $q$ admits linear time evaluation w.r.t.\
\Tmc if and only if $q'$ is acyclic, assuming the hyperclique hypothesis.
\end{corollary}

By virtue of Lemmas~\ref{lem:non-recursive-construction} and~\ref{lem:guarded-friendly}, \Cref{thm:testing-dichotomies} and \Cref{cor:non-rec-Bool} apply to (1)~non-recursive sets of TGDs in which the arity of head atoms is at most two, (2)~non-recursive sets of TGDs in which every TGD contains at most two frontier variables, and (3)~sets of frontier-guarded full TGDs. 
It is an interesting question whether 
\Cref{thm:testing-dichotomies} and \Cref{cor:non-rec-Bool} still apply to non-recursive sets of TGDs when neither the arity of head relations nor the number of frontier variables is restricted.
Note that, in Lemma~\ref{lem:non-recursive-construction}, the tagging companion is a minimized query  and \Cref{thm:testing-dichotomies} and \Cref{cor:non-rec-Bool} refer to \emph{any} tagging companion. 
As we see next, this overall strategy can no longer
work in the unrestricted case.
\begin{example}
    Recall the Boolean CQ $\qtri$ and the TGD $\Ttriguard$ defined in \Cref{sec:pre-discussion}. Further consider the acyclic CQ
    $\qtriguard()\leftarrow R_1(x_1,x_2),R_2(x_2,x_3),R_3(x_3,x_1),S(x_1,x_2,x_3)$. 
    Then $\qtri$ can be evaluated in linear time w.r.t.\ $\Ttriguard$ because $\qtri \equiv_{\Ttriguard} \qtriguard$ and $\qtriguard$ is acyclic.
    However, the minimization of $\qtri$ w.r.t.~$\Ttriguard$ is $\qtri$ itself, and it is cyclic.
\end{example}

This  suggests that a characterization for the unrestricted case has to be more involved. We next show  an example for which the complexity is open, currently preventing us from reaching a full classification for non-recursive sets of TGDs when the arity is unrestricted.
\begin{example}
    For $\ell\geq 3$,  consider the Boolean CQ $q_\ell$ for detecting $\ell$-cliques and
    the TGD $T_\ell$ that says that every $(\ell-1)$-clique has a guard:
   $$
   q_\ell=\bigwedge_{1\leq i < j\leq\ell}R_{i,j}(y_i,y_j)
   \qquad\text{ and } \qquad
   T_\ell = \bigwedge_{1 \leq i < j\leq\ell-1}R_{i,j}(x_i,x_j)
   \rightarrow S(x_1,\dots,x_{\ell-1}).
   $$
    For any $\ell \geq 3$, $q_\ell$ does not admit linear time evaluation 
    on unrestricted databases assuming the triangle hypothesis~\cite{bagan-enum-cdlin}. For $q_3$,
    the same lower bound holds with respect to $T_3$: in the linear time reduction from triangle detection, one uses the same reduction database $D$ as in the case without TGDs, extended with all  facts $S(c_1,c_2)$ for $c_1,c_2 \in \dom(D)$. The reduction still works since this $S$-flooding only takes $O(n^2)$ time. The same argument does not work for $q_4$ and $T_4$ because in that case, $S$ is ternary. However, we can replace
    triangle detection with $4$-clique detection, which is believed to
    not be possible in $O(n^3)$ time \cite{lincoln2018tight}. It is not hard
    to give a linear-time reduction from 4-clique detection 
    to the evaluation of $q_4$ on
    unrestricted databases, and this
    reduction extends to evaluation
    of  $q_4$ w.r.t.\ $T_4$ via
    $S$-flooding.
  This argument, however, does not 
  extend to all $\ell$. In particular,
    for sufficiently large~$\ell$ it is possible to detect $\ell$-cliques in $O(n^{\ell-1})$ time using matrix multiplication~\cite{eisenbrand2004cliques}.     
    This situation resembles an open case of UCQ evaluation, see~\cite[Example 36]{carmeli2021ucqs}.
\end{example}

\section{Counting}
\label{sect:counting}

\emph{Counting the answers to a CQ $q$ w.r.t.\ a set of TGDs} \Tmc, both over the same schema \Sbf, is the
problem to compute, given as input an \Sbf-database $D$ that satisfies all TGDs in \Tmc, the cardinality of
$q(D)$. We start with stating the known dichotomy for the case without TGDs. 
It pertains to the exponential time hypothesis.

Consider the satisfiability problem for propositional logic formulas in conjunctive normal form with $m$ variables and $k$ variables per clause ($k$-SAT),
and let $s_k$ denote the infimum of the real numbers $\delta$
for which this problem admits an $O(2^{\delta m})$ algorithm. The Strong Exponential Time Hypothesis (SETH) states that $\lim_{k \to \infty} s_k = 1$.
If a CQ $q$ is acyclic and free-connex, then it admits linear time counting, see the recent survey~\cite{DBLP:journals/corr/abs-2506-17702}. If $q$ is self-join-free and cyclic, then it does not admit counting in $O(|D|+|\dom(D)|^2)$ time, assuming the hyperclique hypothesis~\cite{BraultBaron}.
Moreover, if $q$ is self-join-free, acyclic, and not free-connex, then it does not admit linear time counting, assuming SETH~\cite{mengel-counting-note}. 

We lift these results
to TGDs, starting with a central lemma that connects answer counting for a query w.r.t.\ a set of TGDs $T$ to answer counting for the self-join free version of its tagging companion on unrestricted databases. To prove the lemma, we make use of a coloring
technique introduced in~\cite{chen-counting}.
For a CQ $q(\bar x)$, we use
$q^*(\bar x)$ to denote the result of 
coloring every variable with a different color. More precisely, we 
add
an atom $A_x(x)$ for every variable $x$ in $q$, with $A_x$  a fresh unary relation symbol.
Chen and Mengel~\cite{chen-counting} give a linear time (Turing) reduction from the counting problem for $q(\bar x)$ to that for $q^*(\bar x)$, without TGDs. We observe that it also applies
in the presence of a set of full TGDs.
\begin{theorem}~\cite[Lemma 30]{chen-counting}\label{thm:chen-counting}
    Let $q(\bar x)$ be a CQ that is a core,
    and let \Tmc be a set of full TGDs satisfied by $D_q$. If $q(\bar x)$
    admits linear time counting w.r.t.\ \Tmc, then so does $q^*(\bar x)$.
\end{theorem}
\Cref{thm:chen-counting} is in fact a reformulation of Chen and Mengel's Lemma~30 for our purposes. Chen and Mengel do not explicitly speak about linear time, but their proof indeed yields a linear time algorithm. More precisely, the algorithm constructs databases obtained from the input database $D$
by first constructing, in the spirit
of Definition~\ref{def:tagging},
$$
\begin{array}{rcl}
  D' &=& \{ R(\langle y_1,c_1 \rangle ,\dots,\langle y_\ell,c_\ell\rangle ) \mid R(y_1,\dots,y_\ell) \in q, R(c_1,\dots,c_\ell) \in D,
  \text{ and } \\[1mm]
  && \hspace*{3.75cm} A_{y_i}(c_i) \in D
  \text{ for all $i$ with } 1 \leq i \leq \ell\},
\end{array}
$$
and then adding a linear number of clones of existing elements.
The answer is obtained by solving an equation system in which the number of equations and the number of unknowns
depend only on the size of the query, and whose coefficients are obtained by counting answers on the constructed databases.
This is clearly possible in linear time.
It is not hard to see that if
$D$ satisfies all TGDs in \Tmc, then so does $D'$: it is an induced subdatabase of the product of $D$ and $D_q$, full TGDs are preserved under taking induced subdatabases, and taking the product preserves all TGDs that are true in both factors.
Moreover, cloning elements also preserves the satisfaction of TGDs. See also~\cite{LMCS-counting-23}, where these observations were exploited.

We use \Cref{thm:chen-counting} to prove the previously announced central lemma. 
\begin{lemma}\label{lem:count-reduction}
    Let \Tmc be a finite set of full TGDs that is tagging friendly with time bound $t_\mn{tag}$. Further, let $q(\bar x)$ be a CQ over the same schema as \Tmc and let $q'(\bar x)$ be its tagging companion.
    If  $q$ admits linear time counting w.r.t.\ \Tmc, then there exists an $O(|D|+t_\mn{tag})$ time algorithm for  counting answers to $q'_{\mn{sjf}}$ over unrestricted databases~$D$.
\end{lemma}
\begin{proof}
    Assume that we are given a database $D$ in the schema of $q'_{\mn{sjf}}$. Since $q'$ is a tagging companion of $q$, we can construct a database $\Dtag$ for $q$  satisfying all TGDs in $\Tmc$ such that $\bar c \in q'_{\mn{sjf}}(D)$ if and only if $\bar x \otimes \bar c \in q(\Dtag)$ for all $\bar c \in \dom(D)^{|\bar x|}$.
    By Lemma~\ref{lem:chasingqueryequiv},
    chasing $q$ with \Tmc results in a query $\mn{ch}_\Tmc(q)$ such that,
     $\mn{ch}_\Tmc(q)(\Dtag)=q(\Dtag)$.
    By \cref{claim:cored-chase}, the query $\core(\mn{ch}_\Tmc(q))$ also satisfies all TGDs in \Tmc, and in addition we have $\core(\mn{ch}_\Tmc(q))(\Dtag)=q(\Dtag)$.
    To compute $|q'_{\mn{sjf}}(D)|$, it is thus enough to count those answers to $\core(\mn{ch}_\Tmc(q))$
    on $\Dtag$ that are of the form
    $\bar x \otimes \bar c$, for some $\bar c \in \dom(D)^{|\bar x|}$.
    Extend $\Dtag$ to a database $\Dtag^*$ in the schema of the
    colored query $\core(\mn{ch}_\Tmc(q))^*$ as follows. For every answer variable~$x$, add the fact $A_x(\langle x, c\rangle)$ for every  $c \in \dom(D)$. For every quantified variable $y$, add the fact $A_y(\langle z,c\rangle )$ for
    every $(z,c) \in \dom(\Dtag)$. Then,  $\core(\mn{ch}_\Tmc(q))^*(\Dtag^*)$ is exactly the set of those answers to $\core(\mn{ch}_\Tmc(q))$ on $\Dtag$ that are of the form $\bar x \otimes \bar c$. Thus, $|\core(\mn{ch}_\Tmc(q))^*(\Dtag^*)|=|q'_{\mn{sjf}}(D)|$.    
    A counting algorithm for $q$ w.r.t.\ \Tmc functions as a counting algorithm for the equivalent query $\core(\mn{ch}_\Tmc(q))$ w.r.t.\ \Tmc, which by \Cref{thm:chen-counting} implies such an algorithm for $\core(\mn{ch}_\Tmc(q))^*$. 
    The construction of $\Dtag$ takes $O(t_\mn{tag})$ time, extending it to $\Dtag^*$ takes linear time, and counting $|\core(\mn{ch}_\Tmc(q))^*(\Dtag^*)|$ takes $O(|\Dtag^*|)\leq O(|\Dtag|)\leq O(|D|+t_\mn{tag})$ time.
\end{proof}

To obtain dichotomies, we first consider the case of full CQs.
Note that every full CQ  is free-connex, and thus in the case without TGDs the only reason why a (self-join free) full CQ  may not admit linear time counting is that 
it is acyclic. In fact, it then does not even admit counting in $O(|D|+|\dom(D)|^2)$ time~\cite{BraultBaron}, assuming the hyperclique hypothesis.
The following theorem, which is easy to obtain from Lemma~\ref{lem:count-reduction}, lifts this result to CQs that may contain self-joins, in the presence of full TGDs.
\begin{thmrep}\label{thm:non-rec-counting}
    Let \Tmc be a finite  set of full TGDs that is tagging friendly with time bound $O(|D|+|\dom(D)|^2)$. Further, let $q(\bar x)$ be a full CQ over the same schema as~\Tmc, and let $q'(\bar x)$ be its tagging companion. Then,
        $q$ admits linear time counting w.r.t.~$\Tmc$ if and only if $q'$ is acyclic, assuming the hyperclique hypothesis.
\end{thmrep}
\begin{proof}
    It is known that full CQs without self-joins admit linear time counting over unrestricted databases if and only if they are acyclic, assuming the hyperclique hypothesis~\cite{DBLP:journals/corr/abs-2506-17702}.

    First assume that $q'$ is acyclic. Then, so is $q'_{\mn{sjf}}$, and so $q'_{\mn{sjf}}$ admits counting in linear time over unrestricted databases. Given a database $D$ for $q'$, we can duplicate facts to build a database $D_{\mn{sjf}}$ for $q'_{\mn{sjf}}$  such that $q'(D)=q'_{\mn{sjf}}(D_{\mn{sjf}})$. Thus, $q'$ also admits linear time counting over unrestricted databases. It follows
    that $q$ admits linear time counting w.r.t.\ \Tmc, since $q \equiv_\Tmc q'$.
    
    For the other direction, assume that $q'$ is cyclic. Then, so is $q'_{\mn{sjf}}$, and thus $q'_{\mn{sjf}}$ does not admit linear time counting over unrestricted databases, assuming the hyperclique hypothesis. 
    By \Cref{lem:count-reduction}, a linear time counting algorithm for $q$ w.r.t.\ $\Tmc$ would imply an $O(|D|+|\dom(D)|^2)$ counting algorithm for $q'_{\mn{sjf}}$ over unrestricted databases, which contradicts the hyperclique hypothesis.
\end{proof}

By  Lemmas~\ref{lem:non-recursive-construction} and~\ref{lem:guarded-friendly}, 
\cref{thm:non-rec-counting}
applies to (1)~non-recursive sets of full TGDs in which the arity of head atoms is at most two, (2)~non-recursive sets of full TGDs in which every TGD contains at most two frontier variables, and (3)~sets of frontier-guarded full TGDs.

We remark that there is no 
obvious way to use 
\Cref{cor:non-rec-Bool} to
prove  the lower bound in \Cref{thm:non-rec-counting}. It might be tempting to 
think that, to prove that a full CQ $q(\bar x)$ does
not admit linear time counting w.r.t.\ a set of TGDs $T$, it suffices to  show
that the Boolean CQ $q_0$ obtained from $q$ by quantifying all variables does not admit linear time evaluation. This is because, clearly, a linear time counting algorithm for $q$ w.r.t.\ \Tmc makes it possible to evaluate $q_0$ in linear time. However, the
tagging companion of
$q$ being cyclic does not imply that the tagging companion of $q_0$ is 
cyclic. As an example, consider the
query
$q(x_1,x_2,x_3,x_4) \leftarrow R_1(x_1,x_2),R_2(x_2,x_3),R_3(x_1,x_4),R_4(x_4,x_3)$ and the empty set of TGDs. As there are no TGDs, the core of $q$ is a  tagging companion of $q$, and likewise for $q_0$. But the core of $q$ is $q$ itself, and thus cyclic, while the core of $q_0$ is a path of length two, thus acyclic.

\medskip

We next consider unrestricted CQs, that is, CQs that may contain quantified variables.
For self-join free CQs without TGDs, there is then an additional reason for why a query $q$ may fail to admit linear time counting:
$q$ could be acyclic but not free-connex. 
While the lower bound for cyclic CQs even establishes  that there is no $O(|D|+|\dom(D)|^2)$ time counting algorithm (assuming the hyperclique hypothesis), this is not the case for the lower bound for CQs that are acyclic but not free-connex.
More precisely, the proof in \cite{mengel-counting-note} that there is no linear time counting for such a CQ $q$  relies on the facts 
that any counting algorithm for $q$ can be used to solve the $k$-dominating set problem and, moreover, that
 SETH implies that there is no $O(n^{k-\varepsilon})$ algorithm that solves $k$-dominating set for any constant $\varepsilon$. However,
if one assumes an $O(|D|+|\dom(D)|^2)$ time counting algorithm for $q$ rather than a linear time one,
the construction from \cite{mengel-counting-note} no longer yields an
algorithm for $k$-dominating set with
 a contradictory running time. 
We illustrate the problem using a
simple example.

\begin{example}
    Consider the following CQ and non-recursive full TGD:
    $$
    q(x_1,x_2)\leftarrow R_1(x_1,y),R_2(y,x_2)
    \qquad
    T= R_1(x_1,x_2),R_2(x_2,x_3)\rightarrow S(x_1,x_3).$$
    Any database $D$ that satisfies $T$
    must contain at least $|q(D)|$ many
    $S$-facts. In the specific case of the databases constructed in \cite{mengel-counting-note}, this requires the
    addition of $n^k$ many $S$-facts.
    The overall algorithm, which involves first constructing $D$ and then running a counting algorithm for $q$ on $D$, thus cannot have running time $O(n^{k-\varepsilon})$ for any $\varepsilon$. In fact, it remains open whether $q$ admits linear time counting w.r.t.\ \Tmc. 
\end{example}

For sets of TGDs that are tagging friendly with time-bound $O(|D|)$, such
as sets of full frontier-guarded TGDs, the counting dichotomy for self-join free CQs over unrestricted databases extends to self-join free CQs and TGDs.
\begin{thmrep}\label{thm:general-linear-counting}
    Let \Tmc be a  finite set of full TGDs that is tagging friendly with time bound $O(|D|)$. Further, let $q(\bar x)$ be a CQ over the same schema as \Tmc, and let $q'(\bar x)$ be its tagging companion.
  Then, $q$ admits linear time counting w.r.t.\ $\Tmc$ if and only if $q'$ is acyclic and free-connex, assuming the hyperclique hypothesis and SETH.
\end{thmrep}

\begin{proof}
    It is known that self-join free CQs admit linear time counting over unrestricted databases if and only if they are acyclic or not free connex, assuming the hyperclique hypothesis and SETH~\cite{mengel-counting-note,BraultBaron}.

  First assume that $q'$ is acyclic
  and free-connex. Then, so is $q'_{\mn{sjf}}$, and thus $q'_{\mn{sjf}}$ admits counting in linear time over unrestricted databases. Given a database $D$ for $q'$, we can duplicate facts to build a database $D_{\mn{sjf}}$ for $q'_{\mn{sjf}}$  such that $q'(D)=q'_{\mn{sjf}}(D_{\mn{sjf}})$. Thus, $q'$ also admits linear time counting over unrestricted databases.
  It follows that $q$ admits linear time counting w.r.t.\ \Tmc, since $q \equiv_\Tmc q'$.

        For the other direction, assume that $q'$ is cyclic or not free-connex. Then, so is $q'_{\mn{sjf}}$, and thus $q'_{\mn{sjf}}$ does not admit linear time counting over unrestricted databases, assuming the hyperclique hypothesis. 
    By \Cref{lem:count-reduction} with time bound $t_\mn{tag}=O(|D|)$, a linear time counting algorithm for $q$ w.r.t.\ $\Tmc$ would imply a linear time counting algorithm for $q'_{\mn{sjf}}$ over unrestricted databases, which contradicts the hyperclique hypothesis.
\end{proof}

\section{Lexicographic Direct Access}

\emph{Lexicographic direct access} for a query $q$ w.r.t.\ an ordering $L$ of its answer variables and a set of TGDs \Tmc, over the same schema \Sbf as $q$, is the following problem. The initial input is an \Sbf-database $D$ that satisfies all TGDs in \Tmc and
there is a \emph{preprocessing phase} in which the algorithm may produce suitable data structures, but no output. In the subsequent \emph{access phase}, the algorithm repeatedly receives a positive integer $i$ and returns the $i$th answer in $q(D)$ sorted lexicographically according to $L$; if $i>|q(D)|$, an out-of-bound signal is returned. The \emph{preprocessing time} of the algorithm is the time spent in the preprocessing phase, while the \emph{access time} is the time it takes the algorithm to output an answer given its index.

We recall the known dichotomy for the case without TGDs. 
For a CQ $q$ and an ordering $L$ of some of its answer variables, a \emph{disruptive $L$-trio} is a triple $(v_1,v_2,v_3)$ such that $v_1$ and $v_2$ appear before $v_3$  in $L$, no atom in $q$ contains both $v_1$ and $v_2$, and $v_3$ shares in $q$ an atom with $v_1$ and an atom with $v_2$.
 If a CQ $q$ is acyclic and free-connex and contains no disruptive $L$-trio, then it admits direct access with linear preprocessing and logarithmic access time~\cite{carmeli2023da}.
We define the \loghyperclique{} hypothesis in the same way as the hyperclique hypothesis, but allowing logarithmic factors: for all $k\geq 3$, it is not possible to determine the existence of a $k$-hyperclique in a $(k-1)$-uniform hypergraph with $n$ vertices in time $O(n^{k-1}\text{polylog}(n))$.
The following is a slight
variation of Theorem 3.3
in \cite{carmeli2023da}.
\begin{lemmarep}\label{lem:known-da-hardness}
    Assuming the \loghyperclique{} hypothesis, if a self-join free 
    CQ $q$ is not acyclic and free-connex or has a disruptive $L$-trio, then $q$ does not admit lexicographic direct access w.r.t.\ $L$ with $O(|D|+|\dom(D)|^2)$ preprocessing and polylogarithmic access time.
\end{lemmarep}
\begin{proof}
The authors of \cite{carmeli2023da} do not state the result exactly as we do, but their proof given also establishes the statement given here. In particular, they are not interested in the $|\dom(D)|^2$ term, and consequently they use a sparse version of the \loghyperclique{} hypothesis. Given a hypergraph with $n$ vertices, their reduction constructs a database with a domain of size~$n$. Thus $|\dom(D)|^2 \leq n^2$, which is within the $O(n^{k-1}\text{polylog}(n))$ bound assumed by the hypothesis for all $k\geq 3$.
\end{proof}
We lift this result to the case with tagging-friendly TGDs and with self-joins, 
using an approach that is
similar to that used for counting  in Section~\ref{sect:counting}.
In fact, we proceed via
a suitable form of counting
that pertains to an ordering $L$ of the answer variables in the query, called \emph{counting under prefix constraints}, and use the
results  for this
querying mode established in 
\cite{bringmann2025tight}. The details are in the appendix, and we only state the main result.

\begin{toappendix}
We define the task of counting under prefix constraints.
Consider the following notion of constraints on query answers. Let $q$ be a CQ, $L$ an order of
its answer variables, and $D$ a database. We define a \emph{prefix constraint} $c$ on a prefix
$x_1,\ldots,x_r$ of $L$ to be a mapping with domain $\{x_1,\ldots,x_r\}$ such that $c(x_i) \in \dom(D)$ for $1 \leq i < r$, and $c(x_r) \subseteq \dom(D)$ is an interval, that is, the set $c(x_r)$  can be ordered into a contiguous sequence under $L$. 
To make the
notation for prefix constraints more symmetric, we treat elements of $\dom(D)$ as intervals of
length~$1$, so $c$ maps all $x_i$ in the prefix to intervals, but for $1 \leq i < r$ these intervals
have length~$1$. An answer $a \in q(D)$ \emph{satisfies}  prefix constraint $c$ if $a(x_i) \in c(x_i)$  for  $1 \leq i \leq r$. 

\emph{Counting under prefix constraints}  for a CQ $q$, endowed with an ordering $L$ of its answer variables, w.r.t.\ a set of TGDs \Tmc, both over the same schema \Sbf, is the following problem. The initial input is an \Sbf-database $D$ that satisfies all TGDs in \Tmc.
There is a \emph{preprocessing phase} in which the algorithm may produce suitable data structures, but no output. In the subsequent \emph{counting phase}, the algorithm repeatedly receives a prefix constraint $c$ and returns the number of  answers in $q(D)$ that satisfy $c$. The \emph{preprocessing time} of the algorithm is the time spent in the preprocessing phase, while the \emph{access time} is the time it takes the algorithm to output a count given a constraint.

The following is a restatement of results by Bringmann et al.~\cite{bringmann2025tight} for general CQs instead of full CQs and restricted to a set of TGDs.

\begin{lemma}~\cite[Lemma 36]{bringmann2025tight}\label{lem:prefix-count-to-colored}
    Let $q$ be a CQ that is a core, $L$ an ordering of its answer variables, and \Tmc a set of TGDs satisfied by $D_q$. If there is an algorithm for counting under prefix constraints for $q$ and $L$ on databases satisfying \Tmc, then there is an algorithm for counting under prefix constraints for $q^*$ and $L$ on databases satisfying \Tmc with the same time guarantees.
\end{lemma}
\begin{proof}

As the original proof in~\cite[Lemma 36]{bringmann2025tight} is long, we describe here only the differences required in order to modify the statement to our version: going from {full CQs} to~{general} CQs, and preserving the satisfaction of a set of TGDs.

We first address the transition from {full CQs} to~{general} CQs.
\cite[Lemma 36]{bringmann2025tight} is an extension of \cite[Lemma 30]{chen-counting} from counting all the answers to counting under prefix constraints. However, while \cite[Lemma 30]{chen-counting} handles CQs in general, \cite[Lemma 36]{bringmann2025tight} is restricted to {full CQs} (as these were the only queries considered in that paper). We simply need to reincorporate the support of quantified variables into \cite[Lemma 36]{bringmann2025tight}. The proof for the statement we need follows along the same lines as that of \cite[Lemma 36]{bringmann2025tight}, except instead of counting directly homomorphisms from the query to the database, we count functions from the answer variables to the database domain that can be extended into a homomorphism.
We detail exactly how the definitions given in the proof change next, and the proof remains the same but with the modified definitions.

To match the notation of \cite{bringmann2025tight}, set $Q=q$ and $A_Q=D_q$.
Given a database $D$ for $Q$ and a prefix constraint $c$, we let $hom(A_Q, D, c)$ denote the set of functions from the answer variables of $Q$ to the domain of $D$ that can be extended to a homomorphism from $A_Q$ to $D$ which satisfy $c$.
The following definitions assume that the domain of $D$ is comprised of pairs and that $\pi_i$ denotes the projection to~$i$th component.
We denote by $hom(A_Q, D, c)^{id}$ the set of functions from the answer variables of $Q$ to the domain of $D$ that can be extended to a homomorphism $h'$ from $A_Q$ to $D$ for which $\pi_1\circ h'$ is the identity on the answer variables of $Q$, and $\pi_2\circ h'$ satisfies $c$.
Let $aut(A_Q, c)$ denote the set of functions within the answer variables of $Q$ that can be extended to an automorphism of $A_Q$ that is the identity on all variables of $c$.
Let $hom(A_Q, D, c)^{aut}$ be the set of functions from the answer variables of $Q$ to the domain of $D$ that can be extended to a homomorphism $h'$ from $A_Q$ to $D$ such that $\pi_1\circ h'$ is an automorphism of $A_Q$ that is the identity on all variables of $c$, and $\pi_2\circ h'$ satisfies $c$.
Let $\mn{avar}(Q)$ denote the set of 
answer variables in $Q$.
We define $N_T$ and $N_{T,i}$ as in \cite{bringmann2025tight} for every $T\subseteq \mn{avar}(Q)$, and apply the inclusion-exclusion principle over the subsets of answer variables.
We define the new databases $D_{T,j}$ for every $j\in[\mn{avar}(Q)+1]$.
The equivalent of Observation $39$~\cite{bringmann2025tight} states that, for every homomorphism $h$ from $Q$ to $D$, since $Q$ is a core, $(\pi_1\circ h)(\mn{avar}(Q))=\mn{avar}(Q)$ iff $\pi_1\circ h$ is an automorphism.

{Similarly to the counting case, it is left to discuss the satisfaction of TGDs. 
Recall that the relations $A_x$, defined for every variable $x$, are added in the definition of $q^*$, so they do not appear in the TGDs.
As an input, we are given a database $D^c$ for $q^*$ satisfying the TGDs. 
The database $D$ for $q$ is built by~removing the unary relations $A_x$ for every variable $x$, tagging the remaining relations according to~$q^*$, and taking an induced subdatabase. Since the TGDs do not use the relations $A_x$, the database $D^c$ without the unary atoms satisfies them too.
Since $q^*$ and $D^c$ without the unary atoms  satisfy the TGDs, so does their product~$D'$. Finally, the database $D$ satisfies the TGD because it is an~induced subdatabase of $D'$ and the TGDs are full.}
Next, $D_{T,j}$ is~built from~$D$ given a~subset $T$ of the answer variables and a number $j$ as follows. Each tag~$x$ of a variable $x\in T$ is refined into $j$ distinct tags $x^{(k)}$ with $k\in [j]$, while each appearance of~a~value tagged with $x$ in a fact causes this fact to be copied $j$ times to appear with all possible refinements of~this tag. It is easy to see that if $D$ satisfies a set of TGDs, then $D_{T,j}$ satisfies them too.
\end{proof}

The following is a variation of Proposition 35 in \cite {bringmann2025tight}. While that proposition does not mention TGDs, the proof works as is.
\begin{lemma}\label{lem:da-prefix-count}
    Let $q$ be a CQ, $L$ an ordering of its answer variables, and \Tmc a set of TGDs. Then, $q$ admits  direct access w.r.t.\  $L$ and \Tmc  with linear preprocessing time and polylogarithmic access time  if and only if $q$ admits counting under prefix constraints w.r.t.\ $L$ and \Tmc with linear preprocessing time and polylogarithmic counting time.
\end{lemma}
\begin{proof}
The proof of \cite[Proposition 35]{bringmann2025tight}  uses straightforward binary search to go between the two tasks. It applies as is to prove  \Cref{lem:da-prefix-count} because it does not depend on the form of the query (and in particular it applies to CQs, not only to full CQs as studied in \cite{bringmann2025tight}) and it does not modify the database (so it holds also w.r.t.\ TGDs).
\end{proof}
\end{toappendix}

\begin{thmrep}
\label{thm:directaccessmain}
    Let \Tmc be a  finite set of full TGDs  that is tagging friendly with time bound $O(|D|+|\dom(D)|^2)$. Further, let $q(\bar x)$ be a  CQ  over the same schema as \Tmc, let $q'(\bar x)$ be its tagging companion, and let $L$ be an ordering of the answer variables of $q$.
    Then, $q$ admits direct access w.r.t.\ $L$ and \Tmc with linear preprocessing time and polylogarithmic access time if and only if $q'$ is acyclic and free-connex and contains no disruptive $L$-trios, assuming the \loghyperclique{} hypothesis.
\end{thmrep}
\begin{proof}
    The first direction is trivial: if $q'$ is acyclic free-connex and contains no disruptive $L$-trios, then it admits direct access with linear preprocessing and logarithmic access time even if we ignore \Tmc~\cite{carmeli2023da}. Since $q$ and $q'$ are equivalent under \Tmc, the algorithm for $q'$ can serve as an algorithm for $q$ over databases that satisfy \Tmc. It is left to prove the opposite direction.

    Assume that $q$ admits direct access with linear preprocessing time and polylogarithmic access time on databases satisfying $\Tmc$ with the lexicographic order given by $L$. By \Cref{lem:da-prefix-count}, $q$ and $L$ also admit counting under prefix constraints with linear preprocessing time and polylogarithmic counting time on databases satisfying $\Tmc$. 
    On such database, $q$ and $\core(\mn{ch}_\Tmc(q))$ are equivalent. 
    By \Cref{lem:prefix-count-to-colored}, we get efficient counting under prefix constraints also for $\core(\mn{ch}_\Tmc(q))^*$. Using \Cref{lem:da-prefix-count} again, $\core(\mn{ch}_\Tmc(q))^*$ and $L$ admit direct access with linear preprocessing time and polylogarithmic access time on databases satisfying $\Tmc$.
 
    Let $q'_\mn{sjf}$ be a self-join-free version of $q'$. 
    Since $q'$ is a tagging companion of $q$, we can construct a database $\Dtag$ for $q$ satisfying $\Tmc$ such that $\bar c \in q'_{\mn{sjf}}(D)$ if and only if $\bar x \otimes \bar c \in q(\Dtag)$.
    Since $\Dtag$ satisfies \Tmc, $q(\Dtag)=\mn{ch}_\Tmc(q)(\Dtag)=
    core(\mn{ch}_\Tmc(q)(\Dtag))$.
    We now want to use $\core(\mn{ch}_\Tmc(q))^*$ to filter the answers over $\Dtag$ to only contain those with the identity tags on the answer variables. Extend $\Dtag$ into a database $\Dtag^*$ for $\core(\mn{ch}_\Tmc(q))^*$ as follows. For every answer variable $x$, define the relation $A_x$ to contain the tuple $(\langle x, c\rangle)$ for every value $c$ in the domain of $D$. For every quantified variable $y$, define the relation $A_y$ to contain all values in the domain of $\Dtag$. Using this construction, the answers $\core(\mn{ch}_\Tmc(q))^*(\Dtag^*)$ are exactly those answers of $\core(\mn{ch}_\Tmc(q))(\Dtag)$ that are of the form $\bar x \otimes \bar c$. Thus, to access the answers $q'_\mn{sjf}(D)$, we can access the answers $\core(\mn{ch}_\Tmc(q))^*(\Dtag^*)$ and project every domain value $\langle x, c\rangle$ into $c$.
    Overall, a direct access algorithm for $\core(\mn{ch}_\Tmc(q))^*$ and $L$ over databases satisfying $\Tmc$ with linear preprocessing time and polylogarithmic access time would imply a direct access algorithm for $q'_\mn{sjf}$ and $L$ over unrestricted databases with $O(|D|+|\dom(D)|^2)$ preprocessing time and polylogarithmic access time. This contradicts the \loghyperclique{} hypothesis according to \Cref{lem:known-da-hardness}.
\end{proof}
\cref{thm:directaccessmain}
applies to (1)~non-recursive sets of full TGDs in which the arity of head atoms is at most two, (2)~non-recursive sets of full TGDs in which every TGD contains at most two frontier variables, and (3)~sets of frontier-guarded full TGDs.

\section{Enumeration}

\emph{Enumerating} the answers to a query $q$ w.r.t.\ a set of TGDs \Tmc, both over the same schema \Sbf, 
is the following problem. The only
input is an \Sbf-database $D$ that 
satisfies all TGDs in \Tmc, and there
is a \emph{preprocessing phase} and an \emph{enumeration phase}. The purpose
of the former is to build data structures while not producing any output. In the latter, the algorithm
outputs  all query answers, one by one in unspecified order and without repetition, followed by an end-of-output signal. The \emph{preprocessing
time} on an enumeration algorithm is the time spent in the preprocessing phase, and the \emph{delay} is the time spent between two consecutive answers.

We recall the known dichotomy for the case without TGDs. 
If a CQ $q$ is acyclic and free-connex, then it admits enumeration with linear preprocessing and constant delay. Otherwise and if $q$ is self-join-free, then it does not admit such enumeration, assuming the hyperclique hypothesis~\cite{bagan-enum-cdlin,BraultBaron}.
For CQs with self-joins, a query can admit efficient enumeration even if it is not acyclic and free-connex, but a full characterization is not available and there are concrete CQs whose enumeration complexity is open~\cite{carmeli-enum-sjs}.

The challenges arising from self-joins  hint at  the fact that enumeration is harder to analyze than the other evaluation modes considered in this paper, both
with and without TGDs. In previous sections we 
were able to show that self-joins do not change the classification, intuitively because the tagging technique is compatible with the evaluation modes  studied there. In all-testing, for example,  we can 
test only for answers with the desired tag,
whereas in  enumeration we cannot prevent the output of ill-tagged answers.
Avoiding the hard problem of classifying 
the enumeration complexity of CQs with self-joins,
one may hope to classify the enumeration complexity
of self-join free CQs w.r.t.\ TGDs. This, however,
is still too optimistic as very simple sets of TGDs, and in particular tagging-friendly ones, make it 
possible to recover CQs with self-joins. Let us
make this more precise.  With a \emph{relation inclusion}, we mean a TGD of the form $R_1(x_1,\dots,x_n)\rightarrow R_2(x_1,\dots,x_n)$, which for 
brevity we denote by $R_1\inc R_2$.  Relation inclusions are frontier-guarded full TGDs, and thus sets of relation
inclusions are  tagging-friendly. It is not too
hard to see that (recursive) sets of relation inclusions can be used to
simulate self-joins in the query, details are in the
appendix. 

\begin{toappendix}
    We argue that a classification  of all self-join
free CQs w.r.t.\ sets of relation inclusions 
according to whether or not they admit enumeration with linear preprocessing and constant delay  yields such a classification of all CQs, also with self-joins. 

Let $q$ be a CQ over some schema \Sbf, possibly with self-joins, and consider its self-join free version $q_{\mn{sjf}}$ and the
set \Tmc of relation inclusions that contains $R_{\bar z} \subseteq  R_{\bar z'}$
for all relation symbols $R_{\bar z},  R_{\bar z'}$ in  $q_{\mn{sjf}}$ that are based on the same relation symbol $R$ from \Sbf.  We argue that
$q$ can be enumerated over unrestricted databases with linear preprocessing and constant delay if and only if the same is
true for $q_{\mn{sjf}}$ w.r.t.\ \Tmc.

\smallskip
``if''. 
Assume that $q_{\mn{sjf}}$ can be enumerated w.r.t.\ \Tmc  with linear preprocessing and constant delay.
Then, given an \Sbf-database $D$, we have $q(D)=q_{\mn{sjf}}(D')$
where $D'$ is the database obtained from $D$ by replacing every fact $R(\bar c)$ with the set of all facts $R_{\bar z}(\bar c)$ such that $R_{\bar z}$ is a relation symbol in $q_{\mn{sjf}}$. Clearly, $D'$ satisfies all TGDs in \Tmc and can be constructed in linear time. We can thus use the algorithm for $q_{\mn{sjf}}$ w.r.t.\ \Tmc to enumerate the answers to $q$ on $D$.

\smallskip
``only if''. Assume that $q$  
can be enumerated on unrestricted databases  with linear preprocessing and constant delay.
Assume we are given a database $D$ in the schema of $q_{\mn{sjf}}$ Then $q_{\mn{sjf}}(D)=q(D')$ where $D'$ is the database obtained from $D$ by
replacing every set of facts 
$\{ R_{\bar z_1}(\bar c),\dots,
 R_{\bar z_k}(\bar c) \}$, where
$R_{\bar z_1},\dots,
 R_{\bar z_k}$ are all relation
 symbols in $q_{\mn{sjf}}$ that are based on $R$, with the fact $R(\bar c)$. We can thus use the algorithm for $q$ to enumerate the answers to $q_{\mn{sjf}}$ on $D$.

\end{toappendix}

\subsection{Restricted Cases}\label{sec:enum-low-arity-cqs}

We consider two very restricted cases in which the challenges arising from self-joins are avoided, and a dichotomy can be attained, starting
 with CQs of arity at most two.  \Cref{ex:Bool-cyclic-easy} shows that even in this case, obtaining a dichotomy for enumeration seems very complicated for unrestricted TGDs. We show that tagging-friendly TGDs are more well-behaved.
\begin{thmrep}
\label{thm:aritytwo}
Let \Tmc be a finite set of TGDs that is tagging friendly with time
bound $O(|D|+|\dom(D)|^2)$.
Further, let $q(\bar x)$ be a CQ of arity at most $2$ over the same schema as~\Tmc, and let $q'(\bar x)$ be its tagging companion.
Then, $q$ admits enumeration w.r.t.\ \Tmc with linear preprocessing and constant delay if and only if $q'$ is acyclic and free-connex, assuming the hyperclique hypothesis.
\end{thmrep}
\begin{proof}
    If $q'$ is acyclic and free-connex, then it admits efficient enumeration on unrestricted databases, and since $q$ is equivalent to it, so does $q$. Therefore $q$ in particular admits efficient enumeration w.r.t.~\Tmc.
    Now assume that $q'$ is not acyclic and free-connex. Then, neither is $q'_\mn{sjf}$, and therefore it does not admit efficient enumeration. The
    proof is based on  a construction that uses linear preprocessing and constant delay enumeration for $q'_\mn{sjf}$ to detect $k$-hypercliques in a $(k-1)$-uniform hypergraph with $n$ vertices in time $O(n^{k-1})$, for some $k\ge 3$~\cite{BraultBaron}. In this construction, the  domain of the constructed database is the set of vertices from the input hypergraph, so $|\dom(D)|=n$.
    Given such an input $D$ to $q'_\mn{sjf}$, 
    since $q'$ is a tagging companion of $q$, we can construct a database $\Dtag$ for $q$ satisfying all TGDs in $\Tmc$ such that $\bar c \in q'_{\mn{sjf}}(D)$ if and only if $\bar x \otimes \bar c \in q(D_{\mn{tag}})$. 
    Since $\dom(D_{\mn{tag}}) \subseteq \mn{var}(q') \times \dom(D)$ and the arity of $q$ is at most $2$, we get that $|q(D_{\mn{tag}})|=O(|\dom(D)|^2)=O(n^2)$. We can enumerate $q(D_{\mn{tag}})$ and ignore answers not of the form $\bar x \otimes \bar c$, thus detecting hypercliques in time $O(n^{k-1}+n^2)=O(n^{k-1})$, contradicting the hyperclique hypothesis.
\end{proof}

\cref{thm:aritytwo}
applies to (1)~non-recursive sets of full TGDs in which the arity of head atoms is at most two, (2)~non-recursive sets of full TGDs in which every TGD contains at most two frontier variables, and (3)~sets of frontier-guarded full TGDs.

We next consider TGDs with only unary relation symbols in the head and
show that they do not change the known dichotomy for enumeration for connected self-join free CQs.
\begin{thmrep}
\label{thm:unaryheads}
    Let $\Tmc$ be a finite set of TGDs with only unary relation symbols in the head, and let $q(\bar{x})$ be a connected self-join free CQ over the same schema as \Tmc. Then, the enumeration of $q$ w.r.t.\ \Tmc is possible 
    with linear preprocessing and constant delay  if and only if $q$ is acyclic and free-connex, assuming the hyperclique hypothesis.
\end{thmrep}
\begin{proof}
    If $q$ is acyclic and free-connex, then it admits efficient enumeration on unrestricted databases, and it particular also on databases satisfying $\Tmc$.
    Let $q$ be a self-join-free connected CQ that is not acyclic and free-connex, and let $q'$ be obtained from $q$ by dropping all unary atoms. Then, also $q'$ is  not acyclic free-connex, so it does not admit efficient enumeration on unrestricted databases.
    Given a database $D'$ in the schema of  $q'$, we may build a database $D$ that extends $D'$ with all
    facts $A(c)$ where $A$ is any relation symbol from the schema and $c \in \dom(D')$. Note that $D$ satisfies all possible TGDs in formlated in the relevant schema.
    We claim that $q(D)=q'(D')$, implying that $q$ does not admit efficient enumeration w.r.t.\ $\Tmc$.
    Since $q$ is not acyclic and free-connex, it contains more than one variable, and since it is also connected, every variable in a unary atom in $q$ must also appear in an atom in $q$ of higher arity, and this atom also appears in $q'$. Thus, $q$ and $q'$ contain the same set of variables, and they have the same answer variables.
    Since the only difference between $q$ and $q'$ is that $q$ has more atoms, and since  
    all relations used only in $q$ contain the entire active domain in $D$, we have that $q(D)= q'(D')$. 
\end{proof}
In the appendix, we present an example of a disconnected
 CQ $q$ that is not acyclic and free-connex, and a set of TGDs \Tmc with only unary relation symbols in the head such that
enumerating $q$ w.r.t.\ \Tmc is possible with linear preprocessing and constant delay. Thus, \cref{thm:unaryheads} does not extend to the case of disconnected CQs. In addition, we present an example for which the enumeration complexity remains an open problem.

\begin{toappendix}
The following example presents a CQ $q$ that is not acyclic and free-connex, and a set of TGDs \Tmc with only unary relation symbols in the head such that
enumerating $q$ w.r.t.\ \Tmc is possible with linear preprocessing and constant delay. Thus, \cref{thm:unaryheads} does not extend to the case of disconnected CQs.
\begin{example}\label{ex-enum-unary-disconnected-easy}
Consider the query $q(x_1,x_2,x_3) \leftarrow  R_1(x_1,z),R_2(z,x_2),S(x_3)$ with the TGD
$R_1(v_1,v_2)\rightarrow S(v_2)$.
To evaluate $q$, we define an auxiliary query $q'(y_1,y_2,y_3) \leftarrow  R_1(y_1,y_3),R_2(y_3,y_2)$.
For every answer $(a,c,b)\in q'(D)$, we output $(a,b,c)$ as an answer to $q$ and store $(a,b)$ in an auxiliary relation $R'$. Once we have evaluated $q'$, we simply output $(a,b,c)$ for every $(a,b)\in R'$ and $(c)\in S$.
Due to the TGD, the first output line indeed prints answers to $q$. The second output line computes exactly all answers to $q$, so overall, every answer is found once or twice. Similarly to \Cref{ex-enum-easy-simple}, we can use the Cheater's Lemma~\cite{carmeli2021ucqs} to get linear preprocessing and constant delay enumeration with no duplicates.
\end{example}

The following example exhibits another disconnected CQ and set of TGDs with only unary relation symbols in rule heads, for which the enumeration complexity remains open. It shows that finding a modification of \cref{thm:unaryheads} that applies to all CQs, both connected and disconnected, might be challenging. 
\begin{example}
    $q(x_1,x_2,x_3)\leftarrow R_1(x_1,z_1),R_2(z_1,z_2),R_3(z_2,x_2),S(x_3)$ with the TGDs $R_1(v_1,v_2)\rightarrow S(v_2)$ and $R_2(v_1,v_2)\rightarrow S(v_2)$. 
    The algorithm of \Cref{ex-enum-unary-disconnected-easy} does not work here since there are two quantified variables on the path $x_1-z_1-z_2-x_2$ and only one answer variable $x_3$ propagating their possible values to the answers, so going over answers to $q$ does not give us enough time to go over all assignments to the full path.
    
    It seems unlikely that the answers to this example can be enumerated with linear preprocessing and constant delay, as this would allow detecting triangles in a tripartite graph with vertex sets $|V_1|=|V_2|=n$ and $|V_3|=\sqrt{n}$ in time $O(n^2)$. This does not contradict any well-established hypotheses, but it is not doable by state-of-the-art methods for triangle detection (using matrix multiplication), which do not give a near-quadratic time bound.
    Such a construction can be done by assigning $R_1=E_{3,1}$, $R_2=E_{1,2}$, $R_3=E_{2,3}$, and $S=V_1\cup V_2$. Then, the query returns an answer assigning $x_1$ and $x_2$ the same value iff the graph contains a triangle. A linear preprocessing and constant delay enumeration algorithm would allow producing and checking all answers in time bounded by the input and output sizes, and the number of answers to $q$ is bounded by $|V_3|^2(|V_1|+|V_2|)=O(n^2)$.
\end{example}
\end{toappendix}

\subsection{Challenges}\label{enum-non-rec-role-inc}

We consider binary relation inclusions, referred
to as \emph{role inclusions}, and focus in particular
on non-recursive sets of role inclusions. These are tagging-friendly for two reasons: (i)~they are non-recursive with only two frontier variables
and (ii)~they are full and frontier-guarded.
We show that even though there is 
 no obvious way to simulate CQs with
self-joins, many of the same challenges emerge. 
In particular, there are CQs and sets of TGDs that behave in
unexpected ways, and also cases where the complexity is a non-trivial open problem. 

We start with giving an example of a CQ and a set of non-recursive role inclusions that admit enumeration with linear preprocessing
and constant delay, but that do not fit into the general scheme
identified in the previous sections.
\begin{figure}[t]
    \centering
    \begin{tikzpicture}[scale=.5]

        \node (fake) at (-4,0) {};
        \node (x1) at (4, 4) {$x_1$};
        \node (x2) at (0,2)  {$x_2$};
        \node (x3) at (8,2)  {$x_3$};
        \node (x4) at (4,0)  {$x_4$};
        \node (x5) at (12,2)  {$x_5$};

        \node (x12) at (2.5,3) {};
        \node (x24) at (2,1) {};
        \node (x13) at (5,3) {};
        \node (x34) at (5,1) {};
        \node (x35) at (10,3) {};

        \draw [-latex] (x1)--(x2) node [midway, above] {\scriptsize $R_1$};
        \draw [-latex, very thick] (x4)--(x2) node [midway, below] {\scriptsize $S_1$};
        \draw [-latex, very thick] (x1)to node [midway, above] {\scriptsize $R_2$} (x3);
        \draw [-latex] (x5)--(x3) node [midway, above] {\scriptsize $P$};
        \draw [-latex, very thick] (x4)--(x3) node [midway, below] {\scriptsize $S_2$};

        \draw[dashed, blue, ->] (x35)  arc
        [
            start angle=20,
            end angle=130,
            x radius =2.25cm,
            y radius =1.5cm
        ] ;
        
        \draw[dashed,blue,->] (x13)  arc
        [
            start angle=135,
            delta angle=90,
            x radius =2.0cm,
            y radius =1.3cm
        ] ;

        \draw[dashed,blue,->] (x12)  arc
        [
            start angle=45,
            delta angle=-90,
            x radius =2.0cm,
            y radius =1.3cm
        ] ;

    \end{tikzpicture}
    \caption{The query $q$ from \Cref{ex-enum-easy-simple}, its subquery $q'$ in bold, and an endomorphism on $\mn{ch}_\Tmc(q)$  as dashed arrows (atoms without outgoing arrows map to themselves).}

    \label{fig:ex-enum-easy1}
\vspace*{-2mm}
  \end{figure}
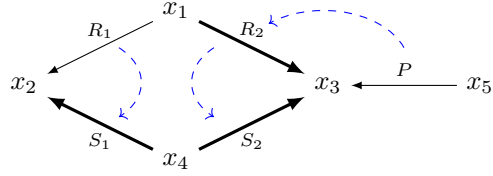
\begin{example}\label{ex-enum-easy-simple}
    Consider the following cyclic CQ and set of TGDs $\Tmc$,  illustrated in \Cref{fig:ex-enum-easy1}.
    $$
    \begin{array}{c}
    q(x_1,x_2,x_3,x_4,x_5)\leftarrow R_1(x_1,x_2),R_2(x_1,x_3),S_1(x_4,x_2),S_2(x_4,x_3),P(x_5,x_3)\\[1mm]
    \Tmc = \{ S_1\inc R_1, S_2\inc R_2, R_2\inc P \}
    \end{array}
    $$
    By \Cref{lem:guarded-friendly},  $\mn{ch}_\Tmc(q)$ (which is a core) is a tagging companion
    of $q$. Since $\mn{ch}_\Tmc(q)$ is  cyclic and in view of the dichotomies given in the preceding sections, one would thus expect $q$ and \Tmc 
    to constitute an intractable case. However, it is possible to enumerate $q$ w.r.t.\ \Tmc with linear preprocessing and
    constant delay, as follows. 
    
    Given a database $D$, enumerate with linear preprocessing and constant delay the answers to the acyclic and free-connex CQ
    $
      q'(x_1,x_2,x_3,x_4)\leftarrow R_2(x_1,x_3),S_1(x_4,x_2),S_2(x_4,x_3),
    $
    displayed by the bold edges in \Cref{fig:ex-enum-easy1}.
    For each answer $(a,b,c,d)$  produced for $q'$, 
    \begin{enumerate}
        \item 
    output $(d,b,c,d,d)$ as an answer to $q$; this is 
    justified since $D$ satisfies \Tmc, cf.\ the endomorphism on $\mn{ch}_\Tmc(q)$  displayed in \Cref{fig:ex-enum-easy1}, with image $\mn{ch}_\Tmc(q')$; 
    
    \item then check whether $(a,b)\in R_1$; if so enumerate all  $(e,c)\in P$, and  output $(a,b,c,d,e)$  as an answer to $q$. 

    \end{enumerate}
    Step~2 may repeat the output
    $(d,b,c,d,d)$ of Step~1. To fix this, we apply the Cheater's Lemma~\cite{carmeli2021ucqs}, which
    states that if every answer is repeated only a constant number of times, then this can be rectified without compromising constant delay.
    It is easy to see that every answer to $q$ is output exactly once in Step~2. Due to Step~1, we output at least one answer to $q$ for every answer to $q'$, which is necessary to achieve constant delay. 
 Testing  $(a,b) \in R_1$ in Step~2 takes constant time on a RAM, and enumerating all $(e,c) \in P$ is possible with constant delay. Overall the algorithm thus uses linear preprocessing and constant delay
\end{example}
In \Cref{ex-enum-easy-simple}, the use of an acyclic and free-connex CQ $q'$ that is the image of an endomorphism on $\mn{ch}_\Tmc(q)$,  is reminiscent of the techniques used
in \cite{carmeli-enum-sjs} to analyze the enumeration complexity
of CQs with self-joins. Using variations of \cref{ex-enum-easy-simple}, we exhibit several challenging phenomena that also arise for CQs with self-joins,
full details are in the appendix.
The variation obtained by dropping the dangling $P$-edge from $q$ is hard (i.e., it does not admit enumeration with linear preprocessing and constant delay) assuming the triangle hypothesis.
If instead of the dangling $P$-edge containing $R_2$, there are dangling $P_i$-edges containing $S_i$ for $i \in \{1,2\}$, the query is still hard, but there is no obvious way to show this assuming the triangle hypothesis; we instead reduce from vertex-unbalanced triangle detection, a problem introduced in \cite{bringmann2025unbalanced}.
Another variation combines two sets of dangling edges with associated inclusions. While each set in isolation
results in a hard query, their combination seems to prevent known hardness proofs, and its complexity remains
open.
 We consider it remarkable that such seemingly minor changes produce such a wide range of behaviors. We also present an example based on a different query that admits enumeration with linear preprocessing and constant delay,
but for reasons that are significantly more intricate than in \Cref{ex-enum-easy-simple}.

\begin{toappendix}
The following example shows that the variation of \Cref{ex-enum-easy-simple} in which the dangling $P$-edge is removed from the query does not admit enumeration with linear preprocessing and constant delay,
assuming the triangle hypothesis.
\begin{example}\label{ex-enum-hard-simple}
Consider the query $q$ and set of TGDs \Tmc obtained from those in \Cref{ex-enum-easy-simple} by dropping the atom and the TGD using $P$.
We show that $q$ does not admit enumeration w.r.t.\ \Tmc with linear preprocessing and
constant delay, assuming that triangle detection in an undirected graph is not possible in time $O(n^2)$.
Let $G=(V,E)$ be the input graph and assume
w.l.o.g.\ that $V \subseteq \mathbb{N}$.
We construct a database $D$ by including, for every edge $(a,b) \in E$ with $a < b$,  the facts 
$$R_1(a,b), R_2(a,b), S_2(a,b), R_1(a,a), S_1(a,a).$$
Note that the TGDs are satisfied.
There are two types of answers to $q$ on $D$:
those of the form $(a,b,c,b)$ with $a<b<c$ and $\{a,b,c\}$ forming a triangle, and those of the form $(a,a,b,a)$ with $a<b$ and $\{a,b\}$ forming an edge. If the answers to $q$ on $D$ can be enumerated with linear preprocessing and constant delay, then after linear time, we either see an answer that exhibits a triangle, or the algorithm finishes printing all answers that correspond to edges and terminates. Thus, such an enumeration algorithm detects the existence of triangles in the input graph in linear time, contradicting the hyperclique hypothesis.
\end{example}

\begin{figure}[t]
    \centering
    \begin{tikzpicture}[yscale=.5]

        \node (x1) at (0, 2) {$x_1$};
        \node (x2) at (-2,0)  {$x_2$};
        \node (x3) at (2,0)  {$x_3$};
        \node (x4) at (0,-2)  {$x_4$};
        \node (x6) at (-4,0)  {$x_6$};
        \node (x5) at (4,0)  {$x_5$};

        \draw [-latex] (x1)--(x2) node [midway, above] {\scriptsize $R_1$};
        \draw [-latex] (x4)--(x2) node [midway, below] {\scriptsize $S_1$};
        \draw [-latex] (x1) to node [midway, above] {\scriptsize $R_2$} (x3);
        \draw [-latex] (x4)--(x3) node [midway, below] {\scriptsize $S_2$};
        \draw [-latex] (x5)--(x3) node [midway, above] {\scriptsize $P_2$};
        \draw [-latex] (x6)--(x2) node [midway, above] {\scriptsize $P_1$};

        \node (x12) at (-.75,.75) {};
        \node (x13) at (.75,.75) {};
        \node (x35) at (3.25,-.5) {};
        \node (x26) at (-3.25,-.5) {};

        \draw[dashed, blue, ->] (x35)  arc
        [
            start angle=0,
            delta angle=-135,
            x radius =1cm,
            y radius =1cm
        ] ;

        \draw[dashed, blue, ->] (x26)  arc
        [
            start angle=180,
            delta angle=135,
            x radius =1cm,
            y radius =1cm
        ] ;
        
        \draw[dashed,blue,->] (x13)  arc
        [
            start angle=140,
            delta angle=80,
            x radius =2.0cm,
            y radius =1.3cm
        ] ;

        \draw[dashed,blue,->] (x12)  arc
        [
            start angle=40,
            delta angle=-80,
            x radius =2.0cm,
            y radius =1.3cm
        ] ;

    \end{tikzpicture}
    \caption{The query from \Cref{ex-enum-hard-utd}.}

    \label{fig:ex-enum-hard-simple}
\vspace*{-2mm}
  \end{figure}

The following variation of Examples~\ref{ex-enum-easy-simple} and~\ref{ex-enum-hard-simple} is also hard, but it does not seem possible to prove this by reduction from triangle detection. 
We instead reduce from vertex-unbalanced
triangle detection, a problem introduced in \cite{bringmann2025unbalanced}.
\begin{example}\label{ex-enum-hard-utd}
    Consider the following query and TGDs, as illustrated in \Cref{fig:ex-enum-hard-simple}.
    \[q(x_1,x_2,x_3,x_4,x_5,x_6)\leftarrow R_1(x_1,x_2),R_2(x_1,x_3),S_1(x_4,x_2),S_2(x_4,x_3),P_2(x_5,x_3),P_1(x_6,x_2)\]
    \[S_1\inc R_1, S_2\inc R_2, S_2\inc P_2, S_1\inc P_1 \]
    The \emph{Vertex-Unbalanced Triangle Detection} (VUTD) hypothesis~\cite{bringmann2025unbalanced} assumes that, for every constant $\alpha\in(0,1]$, it is not possible to decide the existence of a triangle in a tripartite graph with vertex sets $|V_1|=n$, $|V_2|=|V_3|=\Theta(|n^\alpha|)$ in time $O(n^{1+\alpha})$. Note that the input size to this problem is $O(n^{1+\alpha})$. Given an unbalanced tripartite graph $(V_1\cup V_2\cup V_3, E_{1,2}\cup E_{3,2}\cup E_{1,3})$ with $\alpha\le \frac{1}{3}$, we construct a database as follow:
    $S_1=P_1=E_{3,2}$, $R_1=E_{3,2}\cup E_{1,2}$, $S_2=P_2=\{(v,v)\mid v\in V_3\}$, and $R_2=S_2\cup E_{1,3}$.
    There are two types of answers over this construction: $(v_3,v_2,v_3,v_3,v_3,v_3')$ with $v_i,v_i'\in V_i$ and $(v_2,v_3),(v_2,v_3')\in E_{2,3}$; 
    and $(v_1,v_2,v_3,v_3,v_3,v_3')$ with $v_i,v_i'\in V_i$, $\{v_1,v_2,v_3\}$ forming a triangle, and $(v_2,v_3')\in E_{2,3}$.
    If the answers can be enumerated with linear preprocessing and constant delay, after $O(n^{1+\alpha}+|E_{2,3}|^2)=O(n^{1+\alpha}+{(n^{2\alpha})}^2)=O(n^{1+\alpha})$ time, we either see an answer corresponding to a triangle, or the algorithm finishes printing all answers of the first type and terminates. We conclude that this problem does not admit such an efficient enumeration algorithm, assuming the VUTD hypothesis.
    Notice that we cannot use the basic triangle detection assumption instead of VUTD because then there would be too many non-triangle answers to reach a contradiction.
\end{example}

We now present a final variation of 
Examples~\ref{ex-enum-easy-simple} to~\ref{ex-enum-hard-simple}. This time, none of the techniques used to deal with those examples seem to be applicable.
It remains open whether enumeration with linear preprocessing and constant delay is possible.
\begin{example}\label{ex-enum-open}
    The enumeration complexity of the query and TGDs illustrated in \Cref{fig:ex-enum-open} remains open.
    This example is the composition of two hard queries: the query from \Cref{ex-enum-hard-utd}, obtained from this example by removing the relations $T_i$, and the query obtained by removing the relations $P_i$. Each of these queries separately can be proved hard assuming VUTD, but the proof requires assigning the large vertex set to different variables, and so the same proofs do not apply to the combined example. We do not have an algorithm for this example either.
\end{example}

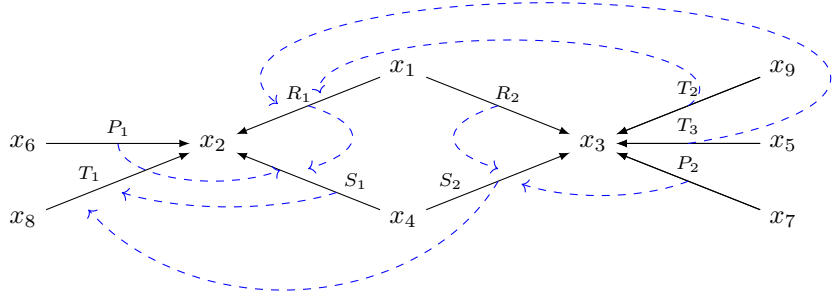
\begin{figure}[t]
    \centering
    \begin{tikzpicture}[yscale=.5, xscale=1.25]

        \node (x1) at (0, 2) {$x_1$};
        \node (x2) at (-2,0)  {$x_2$};
        \node (x3) at (2,0)  {$x_3$};
        \node (x4) at (0,-2)  {$x_4$};
        \node (x5) at (-4,0)  {$x_6$};
        \node (x6) at (-4,-2)  {$x_8$};
        \node (x7) at (4,2)  {$x_9$};
        \node (x8) at (4,0)  {$x_5$};
        \node (x9) at (4,-2)  {$x_7$};

        \draw [-latex] (x1)--(x2);
        \draw [-latex] (x4)--(x2);
        \draw [-latex] (x1) -- (x3);
        \draw [-latex] (x4)--(x3);
        \draw [-latex] (x5)--(x2);
        \draw [-latex] (x6)--(x2);

        \draw [-latex] (x7)--(x3);
        \draw [-latex] (x8)--(x3);
        \draw [-latex] (x9)--(x3);

        \node (r1) at (-1.1, 1.3) {\scriptsize $R_1$};
        \node (r2) at (1.1, 1.3) {\scriptsize $R_2$};

        \node (s1) at (-.5, -1) {\scriptsize $S_1$};
        \node (s2) at (.5, -1) {\scriptsize $S_2$};
        
        \node (l1) at (-3, 0.4) {\scriptsize $P_1$};
        \node (l2) at (-3.3, -.8) {\scriptsize $T_1$};

        \draw [-latex] (x7) to node [midway, above] {\scriptsize $T_2$} (x3);
        \draw [-latex] (x8) to node [midway, above] {\scriptsize $T_3$} (x3);
        \draw [-latex] (x9) to node [midway, above] {\scriptsize $P_2$} (x3);

        \node (x12) at (-1,1) {};
        \node (x13) at (1,1) {};
        \node (x83) at (3,0) {};
        \node (x52) at (-3,0) {};
        \node (x93) at (3,-1) {};
        \node (x73) at (3,1) {};
        \node (x42) at (-.7,-1.3) {};
        \node (x43) at (1,-1) {};
 
        \draw[dashed, blue, ->] (x83)  arc
        [
            start angle=-60,
            delta angle=260,
            x radius =3.0cm,
            y radius =2cm
        ] ;

        \draw[dashed, blue, ->] (x52)  arc
        [
            start angle=180,
            delta angle=135,
            x radius =1cm,
            y radius =1cm
        ] ;
        
        \draw[dashed,blue,->] (x13)  arc
        [
            start angle=140,
            delta angle=80,
            x radius =2.0cm,
            y radius =1.3cm
        ] ;

        \draw[dashed,blue,->] (x12)  arc
        [
            start angle=40,
            delta angle=-80,
            x radius =2.0cm,
            y radius =1.3cm
        ] ;

        \draw[dashed,blue,->] (x93)  arc
        [
            start angle=-45,
            delta angle=-85,
            x radius =1.3cm,
            y radius =1.3cm
        ] ;

        \draw[dashed,blue,->] (x73)  arc
        [
            start angle=-20,
            delta angle=205,
            x radius =2cm,
            y radius =1.3cm
        ] ;

        \draw[dashed,blue,->] (x42)  arc
        [
            start angle=-45,
            delta angle=-90,
            x radius =1.6cm,
            y radius =1.3cm
        ] ;

        \draw[dashed,blue,->] (x43)  arc
        [
            start angle=-25,
            delta angle=-120,
            x radius =2.5cm,
            y radius =5cm
        ] ;

    \end{tikzpicture}
    \caption{The  query from \Cref{ex-enum-open}. A dashed arrow from $A$ to $B$ represents the TGD $B\subseteq A$.}
    \label{fig:ex-enum-open}
\vspace*{-2mm}
  \end{figure}

We close with an adaptation of an example
from \cite{carmeli-enum-sjs}. This is a case that admits enumeration with linear preprocessing and constant delay, but
for reasons that are more intricate than those showcased in \Cref{ex-enum-easy-simple}.

\begin{example}\label{ex-enum-easy-counting}
We can simulate Example 23 from \cite{carmeli-enum-sjs} without self-joins using non-recursive role inclusions. Consider the full CQ $q$ illustrated in \Cref{fig:ex-enum-easy-counting}, where a vertex marked by the number $i$ refers to the variable $x_i$, an edge $i\to j$ describes a binary atom using a unique relation symbol $R_{i,j}$, and a dashed blue arrow from $R$ to $S$ means that $S\inc R$.
We can solve it with the same algorithm given in \cite{carmeli-enum-sjs}, where we use what we mark in \Cref{fig:ex-enum-easy-counting} by $I_s$, $I_\ell$, and $I_t$ as the small, left, and top images respectively.

\begin{figure}[t]
    \centering
    \begin{tikzpicture}[scale=.5]
        
        \node (x20) at (4,0) {12};
        \node (x40) at (8,0) {13};
        \node (x11) at (2,2) {10};
        \node (x31) at (6,2) {5};
        \node (x22) at (4,4) {4};
        \node (x42) at (8,4) {6};
        \node (x02) at (0,4) {9};
        \node (x13) at (2,6) {3};
        \node (x53) at (10,6) {7};
        \node (x24) at (4,8) {2};
        \node (x44) at (8,8) {8};
        \node (x35) at (6,10) {1};
        \node (x55) at (10,10) {14};
        \node (x36) at (6,12) {9};
        \node (nameQ) at (6,6) {$q$};

        \draw [-latex] (x31)--(x20);
        \draw [-latex] (x31)--(x40);
        \draw [-latex] (x11)--(x22);
        \draw [-latex] (x31)--(x22);
        \draw [-latex] (x31)--(x42);
        \draw [-latex] (x02)--(x13);
        \draw [-latex] (x22)--(x13);
        \draw [-latex] (x42)--(x53);
        \draw [-latex] (x24)--(x13);
        \draw [-latex] (x44)--(x53);
        \draw [-latex] (x35)--(x24);
        \draw [-latex] (x35)--(x44);
        \draw [-latex] (x44)--(x55);
        \draw [-latex] (x35)--(x36);

        \node (y512) at (5,1) {};
        \node (y513) at (7,1) {};
        \node (y104) at (3,3) {};
        \node (y54) at (5,3) {};
        \node (y56) at (7,3) {};
        \node (y103) at (1,5) {};
        \node (y43) at (3,5) {};
        \node (y67) at (9,5) {};
        \node (y23) at (3,7) {};
        \node (y87) at (9,7) {};
        \node (y814) at (9,9) {};
        \node (y18) at (7,9) {};
        \node (y12) at (5,9) {};
        \node (y19) at (6,11) {};

        \draw[dashed,blue,->] (y103)  arc
        [
            start angle=-130,
            delta angle=80,
            x radius =1.3cm,
            y radius =1.3cm
        ] ;

        \draw[dashed,blue,->] (y104)  arc
        [
            start angle=-130,
            delta angle=80,
            x radius =1.3cm,
            y radius =1.3cm
        ] ;

        \draw[dashed,blue,->] (y512)  arc
        [
            start angle=-130,
            delta angle=80,
            x radius =1.3cm,
            y radius =1.3cm
        ] ;

        \draw[dashed,blue,->] (y19)  arc
        [
            start angle=90,
            delta angle=-110,
            x radius =1.3cm,
            y radius =1.3cm
        ] ;

        \draw[dashed,blue,->] (y814)  arc
        [
            start angle=40,
            delta angle=-80,
            x radius =1.3cm,
            y radius =1.3cm
        ] ;
        
        \draw[dashed,blue,->] (y513)  arc
        [
            start angle=-40,
            delta angle=80,
            x radius =1.3cm,
            y radius =1.3cm
        ] ;

        \draw[dashed,blue,->] (y56) -- (y54);
        \draw[dashed,blue,->] (y67) -- (y43);
        \draw[dashed,blue,->] (y87) -- (y23);
        \draw[dashed,blue,->] (y18) -- (y12);

        \draw[dashed,blue,->] (y43) -- (y23);
        \draw[dashed,blue,->] (y54) -- (y12);
        \draw[dashed,blue,->] (y56) -- (y18);
        \draw[dashed,blue,->] (y67) -- (y87);

        \node (t20) at ($(2,0)$) {};
        \node (t40) at (4,0) {};
        \node (t11) at (1,1) {};
        \node (t31) at (3,1) {};
        \node (t22) at (2,2) {};
        \node (t42) at (4,2) {};
        \node (t02) at (0,2) {};
        \node (t13) at (1,3) {};
        \node (t53) at (5,3) {};
        \node (t24) at (2,4) {};
        \node (t44) at (4,4) {};
        \node (t35) at (3,5) {};
        \node (t55) at (5,5) {};
        \node (t36) at (3,6) {};

        \node (delta1) at (11,2) {};
        
        \node (z20) at ($1.5*(t20) + (delta1) $) {12};
        \node (z11) at ($1.5*(t11) + (delta1) $) {10};
        \node (z31) at ($1.5*(t31) + (delta1) $) {5};
        \node (z22) at ($1.5*(t22) + (delta1) $) {4};
        \node (z42) at ($1.5*(t42) + (delta1) $) {6};
        \node (z13) at ($1.5*(t13) + (delta1) $) {3};
        \node (z24) at ($1.5*(t24) + (delta1) $) {2};
        \node (z44) at ($1.5*(t44) + (delta1) $) {8};
        \node (z35) at ($1.5*(t35) + (delta1) $) {1};
        \node (nameL) at ($1.5*(3,3) + (delta1) $) {$I_\ell$};

        \draw [-latex] (z31)--(z20);
        \draw [-latex] (z11)--(z22);
        \draw [-latex] (z31)--(z22);
        \draw [-latex] (z31)--(z42);
        \draw [-latex] (z22)--(z13);
        \draw [-latex] (z24)--(z13);
        \draw [-latex] (z35)--(z24);
        \draw [-latex] (z35)--(z44);

        \node (delta1) at (21,2) {};
        
        \node (z22) at ($1.5*(t22) + (delta1) $) {4};
        \node (z13) at ($ 1.5*(t13) + (delta1) $) {3};
        \node (z53) at ($ 1.5*(t53) + (delta1) $) {7};
        \node (z24) at ($ 1.5*(t24) + (delta1) $) {2};
        \node (z44) at ($ 1.5*(t44) + (delta1) $) {8};
        \node (z35) at ($ 1.5*(t35) + (delta1) $) {1};
        \node (z55) at ($ 1.5*(t55) + (delta1) $) {14};     
        \node (nameT) at ($1.5*(3,3) + (delta1) $) {$I_t$};

        \draw [-latex] (z22)--(z13);
        \draw [-latex] (z24)--(z13);
        \draw [-latex] (z44)--(z53);
        \draw [-latex] (z35)--(z24);
        \draw [-latex] (z35)--(z44);
        \draw [-latex] (z44)--(z55);

        \node (delta1) at (16,-3) {};
        
        \node (z22) at ($1.5*(t22) + (delta1) $) {4};
        \node (z13) at ($1.5*(t13) + (delta1) $) {3};
        \node (z24) at ($1.5*(t24) + (delta1) $) {2};
        \node (z44) at ($1.5*(t44) + (delta1) $) {8};
        \node (z35) at ($1.5*(t35) + (delta1) $) {1};     
        \node (nameS) at ($1.5*(3,3) + (delta1) $) {$I_s$};

        \draw [-latex] (z22)--(z13);
        \draw [-latex] (z24)--(z13);
        \draw [-latex] (z35)--(z24);
        \draw [-latex] (z35)--(z44);

    \end{tikzpicture}
    \caption{The query from \Cref{ex-enum-hard-utd}.}

    \label{fig:ex-enum-easy-counting}
\vspace*{-2mm}
  \end{figure}

We repeat the main idea here for completeness. We enumerate the answers to the small image, and, treating each answer $(a,b,c,d,e)$ as an assignment, we show a way to efficiently find all answers to $q$ that extend this assignment. We first find all of the extensions to the top image and print the answers implied by these, while storing the possible assignments to $x_7$ as a set we denote $T_7(e)$. Then, we do the same with the left image while storing the assignments to $x_6$ as a set denoted $T_6(d)$. Next, we go over all pairs $u\in T_6(d)$ and $v\in T_7(e)$ and check whether $(u,v)\in R_{6,7}$. If there is such a fact, we have found an assignment to the cycle $x_1,\ldots,x_8,x_1$, and we can easily enumerate all possible assignments to the spikes $x_9,\ldots,x_{14}$ to get all answers to $q$. The crux of this proof is that the number of pairs of values we check is bounded by the number of preliminary answers we find using the top and left images. Thus, the preliminary answers provide enough time to perform this computation while still achieving linear preprocessing and constant delay.

The same proof from \cite[Example 23]{carmeli-enum-sjs} works here, with the only difference that we need to specify the names of the relations we use (there, all atoms use the same relation). 
We repeat next the computation showing that the number of pairs we check is bounded by the number of preliminary answers.
We define $\inset{a}{b}{c}=\{v|(v,c)\in R_{a, b}\}$ and $\outset{a}{b}{c}=\{v|(c,v)\in R_{a, b}\}$.
We have that $|T_7(e)|\le\outset{8}{7}{e}$ and 
$|T_6(d)|\le\sum_{f\in\inset{5}{4}{d}}{\outset{5}{6}{f}}$, and also 
\begin{align*}
\soltop{e}&=|\outset{8}{7}{e}||\outset{8}{14}{e}|\geq|\outset{8}{7}{e}|^2
\\
\solleft{d}&=|\inset{11}{4}{d}|\sum_{f\in\inset{5}{4}{d}}{|\outset{5}{6}{f}||\outset{6}{12}{f}|}\geq|\inset{5}{4}{d}|\sum_{f\in\inset{5}{4}{d}}{|\outset{5}{6}{f}|^2}
\end{align*}
We get that
\begin{align*}
|T_6(d)|\cdot|T_7(e)|
\le |T_6(d)|^2 + |T_7(e)|^2
\le \left(\sum_{f\in\inset{5}{4}{d}}{\outset{5}{6}{f}}\right)^2 + \left(\outset{8}{7}{e}\right)^2\\
\le
|\inset{5}{4}{d}|\sum_{f\in\inset{5}{4}{d}}{|\outset{5}{6}{f}|^2} + |\outset{8}{7}{e}|^2
\le
\solleft{d} + \soltop{e}
\end{align*}
In other words, in each iteration extending an answer to the small image, the number of pairs we check is bounded by the number of preliminary answers.
\end{example}

\end{toappendix}

\section{Conclusion}

In this article, we have started an investigation of the limits of linear time evaluation of
conjunctive queries under TGD constraints, for the querying modes of single-testing, all-testing, counting, direct access, and enumeration. 
Our results leave open a wealth of interesting questions for future research. The grand goal of completely classifying all CQs w.r.t.\ all sets of TGDs currently seems out of reach, for any of the considered modes of querying. A much more modest aim would be to 
understand the complexity of some of the open cases identified in this article,
hoping that this does not only solve an isolated case, but identifies new techniques.
Another interesting question is whether it is possible to obtain a classification of linear-time evaluation for CQs w.r.t.\ non-recursive sets of TGDs without any further restrictions, and w.r.t.\ frontier-guarded TGDs without assuming that they are full. We  do not know how to apply the tagging approach in this case. However, we are also not aware of examples for which the complexity cannot be determined with current techniques. One idea is to transition from 
the closed world of database constraints to the open world of ontology-mediated querying, as pursued for
instance in \cite{LMCS-counting-23}, and to then rewrite away the existential quantifiers in TGD heads, ending up with full frontier-guarded TGDs as studied in this article. This, however, introduces fresh relation symbols that induce unexpectedly obstinate complications. 

\newpage

\bibliographystyle{plainurl}
\bibliography{local}

\appendix

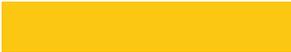
\end{document}